\documentclass[11pt]{article}%
\pdfoutput=1
\usepackage{amsfonts}
\usepackage{amsmath}
\usepackage{fullpage}
\usepackage{amssymb}
\usepackage{graphicx}
\usepackage{hyperref}%
\providecommand{\U}[1]{\protect\rule{.1in}{.1in}}

\newtheorem{theorem}{Theorem}

\newtheorem{lemma}[theorem]{Lemma}

\newtheorem{proposition}[theorem]{Proposition}
\newtheorem{remark}[theorem]{Remark}

\newenvironment{proof}[1][Proof]{\noindent\textbf{#1.} }{\ \rule{0.5em}{0.5em}}
\numberwithin{equation}{section}
\begin{document}

\title{Multipartite quantum correlations and local recoverability}
\author{Mark M. Wilde\thanks{Hearne Institute for Theoretical Physics, Center for
Computation and Technology, Department of Physics and Astronomy, Louisiana
State University, Baton Rouge, Louisiana 70803, USA}}
\maketitle

\begin{abstract}
Characterizing genuine multipartite quantum correlations in quantum physical
systems has historically been a challenging problem in quantum information
theory. More recently however, the \textit{total correlation} or
\textit{multipartite information} measure has been helpful in accomplishing
this goal, especially with the multipartite symmetric quantum (MSQ) discord
[Piani \textit{et al.}, Phys.~Rev.~Lett.~\textbf{100}, 090502, 2008] and the
conditional entanglement of multipartite information (CEMI) [Yang \textit{et
al.}, Phys.~Rev.~Lett.~\textbf{101}, 140501, 2008]. Here we apply a recent and
significant improvement of strong subadditivity of quantum entropy [Fawzi and
Renner, arXiv:1410.0664] in order to develop these quantities further. In
particular, we prove that the MSQ discord is nearly equal to zero if and only
if the multipartite state for which it is evaluated is approximately locally
recoverable after performing measurements on each of its systems. Furthermore,
we prove that the CEMI is a faithful entanglement measure, i.e., it vanishes
if and only if the multipartite state for which it is evaluated is a fully
separable state. Along the way we provide an operational interpretation of the
MSQ discord in terms of the partial state distribution protocol, which in
turn, as a special case, gives an interpretation for the original discord
quantity. Finally, we prove an inequality that could potentially improve upon
the Fawzi-Renner inequality in the multipartite context, but it remains an
open question to determine whether this is so.

\end{abstract}

\section{Introduction}

The quantification and characterization of correlations in multiple physical
systems has a long history, with some of the first proposals for information
measures being the works of McGill \cite{M54} and Watanabe \cite{W60}. Of
particular interest for us here is the \textit{total correlation} measure
proposed by Watanabe \cite{W60}, which is defined for a set of random
variables $X_{1}$, \ldots, $X_{l}$ as the sum of the individual entropies less
the joint entropy:%
\begin{equation}
I\left(  X_{1}:\cdots:X_{l}\right)  \equiv H\left(  X_{1}\right)
+\cdots+H\left(  X_{l}\right)  -H\left(  X_{1}\cdots X_{l}\right)  ,
\end{equation}
where $H\left(  \cdot\right)  $ is the Shannon entropy. The total correlation
has the salient properties of being non-negative and monotone non-increasing
under local operations, meaning that it does not increase under the local
discarding of information, i.e., for random variables $X_{1}$, $X_{1}^{\prime
}$, \ldots, $X_{l}$, $X_{l}^{\prime}$, the following inequality holds%
\begin{equation}
I\left(  X_{1}X_{1}^{\prime}:\cdots:X_{l}X_{l}^{\prime}\right)  \geq I\left(
X_{1}^{\prime}:\cdots:X_{l}^{\prime}\right)  .
\end{equation}

The generalization of the total correlation to quantum physical systems is
straightforward, given simply by replacing Shannon entropies with von Neumann
entropies \cite{H94}. In the quantum information theory literature, the
quantity is known as the \textit{multipartite information}. Specifically, let
$\rho_{A_{1}A_{2}\cdots A_{l}}$ be a multipartite density operator
representing the state of systems $A_{1}$, \ldots, $A_{l}$ (i.e., $\rho
_{A_{1}A_{2}\cdots A_{l}}$ is a trace one, positive semidefinite operator
acting on the tensor-product Hilbert space $\mathcal{H}_{A_{1}}\otimes
\cdots\otimes\mathcal{H}_{A_{l}}$). The multipartite information of this state
is defined as%
\begin{equation}
I\left(  A_{1}:\cdots:A_{l}\right)  _{\rho}\equiv H\left(  A_{1}\right)
_{\rho}+\cdots+H\left(  A_{l}\right)  _{\rho}-H\left(  A_{1}\cdots
A_{l}\right)  _{\rho}, \label{eq:multi-info}%
\end{equation}
with the von Neumann entropy of a density operator $\sigma$ on system $S$
defined in terms of the natural logarithm as $H\left(  S\right)  _{\sigma
}\equiv H\left(  \sigma\right)  \equiv-$Tr$\left\{  \sigma\log\sigma\right\}
$ and the marginal entropies $H\left(  A_{i}\right)  _{\rho}$ are defined with
respect to the reduced density operator%
\begin{equation}
\rho_{A_{i}}=\text{Tr}_{A_{1}\cdots A_{l}\backslash A_{i}}\left\{  \rho
_{A_{1}\cdots A_{l}}\right\}  .
\end{equation}
The quantity in (\ref{eq:multi-info}) is also non-negative and monotone
non-increasing under the local discarding of information, i.e., the following
inequality holds for a multipartite density operator $\rho_{A_{1}A_{1}%
^{\prime}\cdots A_{l}A_{l}^{\prime}}$:%
\begin{equation}
I\left(  A_{1}A_{1}^{\prime}:\cdots:A_{l}A_{l}^{\prime}\right)  _{\rho}\geq
I\left(  A_{1}^{\prime}:\cdots:A_{l}^{\prime}\right)  _{\rho}.
\label{eq:multi-monotone}%
\end{equation}
The above inequality follows because the multipartite information can be
written in terms of the relative entropy $D\left(  \rho\Vert\sigma\right)
\equiv\ $Tr$\left\{  \rho\left[  \log\rho-\log\sigma\right]  \right\}  $
\cite{U62}\ as%
\begin{equation}
I\left(  A_{1}:\cdots:A_{l}\right)  _{\rho}=D\left(  \rho_{A_{1}\cdots A_{l}%
}\Vert\rho_{A_{1}}\otimes\cdots\otimes\rho_{A_{l}}\right)  ,
\end{equation}
and the relative entropy is monotone non-increasing under quantum operations
\cite{U77}, i.e., $D\left(  \rho\Vert\sigma\right)  \geq D\left(
\mathcal{N}\left(  \rho\right)  \Vert\mathcal{N}\left(  \sigma\right)
\right)  $ for any states $\rho$ and $\sigma$ and quantum channel
$\mathcal{N}$ (recall that a quantum channel is a completely positive trace
preserving (CPTP)\ linear map).

Given the inequality in (\ref{eq:multi-monotone}), we are left to wonder
whether one could refine it in a non-trivial way by finding a state-dependent
remainder term. This kind of question has been the driving force behind
several recent investigations in quantum information theory
\cite{Winterconj,K13conj,LW14a,CL14,Z14,BSW14,SBW14,SW14}, culminating in the
following breakthrough inequality of Fawzi and Renner \cite{FR14}:%
\begin{align}
I\left(  A;B|C\right)  _{\rho} &  \geq-\log F\left(  \rho_{ABC},\mathcal{R}%
_{C\rightarrow AC}\left(  \rho_{BC}\right)  \right)  \label{eq:FR}\\
&  \geq\frac{1}{4}\left\Vert \rho_{ABC}-\mathcal{R}_{C\rightarrow AC}\left(
\rho_{BC}\right)  \right\Vert _{1}^{2},
\end{align}
where $I\left(  A;B|C\right)  _{\rho}$ is the conditional quantum mutual
information of a tripartite state $\rho_{ABC}$, defined as%
\begin{equation}
I\left(  A;B|C\right)  _{\rho}\equiv H\left(  AC\right)  _{\rho}+H\left(
BC\right)  _{\rho}-H\left(  C\right)  _{\rho}-H\left(  ABC\right)  _{\rho},
\end{equation}
and $\mathcal{R}_{C\rightarrow AC}$ is a particular CPTP\ \textquotedblleft
recovery map\textquotedblright\ which acts on system $C$ alone in an attempt
to recover the \textquotedblleft lost\textquotedblright\ system$~A$. The
quantity $F\left(  \omega,\tau\right)  \equiv\left\Vert \sqrt{\omega}%
\sqrt{\tau}\right\Vert _{1}^{2}$ is the quantum fidelity between states
$\omega$ and $\tau$ \cite{U73}, with $\left\Vert A\right\Vert _{1}\equiv
\ $Tr$\{\sqrt{A^{\dag}A}\}$ the Schatten $\ell_{1}$ norm. The trace distance
between two density operators $\omega$ and $\tau$ is defined in terms of the
trace norm as $\left\Vert \omega-\tau\right\Vert _{1}$ and characterizes how
well one can distinguish the states $\omega$ and $\tau$ in any physical
experiment. The Fawzi-Renner inequality gives a state-dependent improvement to
strong subadditivity (i.e., $I\left(  A;B|C\right)  _{\rho}\geq0$)
\cite{PhysRevLett.30.434,LR73} and has even been improved upon in recent work
of Brandao \textit{et al}.~\cite{BHOS14}. One can also see the recent work
\cite{BT15}\ for a simpler proof of (\ref{eq:FR}).

The difference of the two multipartite informations in
(\ref{eq:multi-monotone}) is the basis for two distinct measures of quantum
correlations:\ the multipartite symmetric quantum (MSQ)\ discord \cite{PHH08}
and the conditional entanglement of multipartite information
(CEMI)\ \cite{YHW08}, which were inspired by the quantum discord
\cite{Z00,OZ01} and the squashed entanglement \cite{CW04}, respectively. We
briefly motivate these quantities here and give formal definitions later in
the paper. We begin by describing the MSQ\ discord. Let $A_{1}\cdots A_{l}$ be
quantum systems held by spatially separated parties and suppose that each
party measures their local system, leading to classical systems $X_{1}\cdots
X_{l}$. We could then compute the non-negative information gap $I\left(
A_{1}:\cdots:A_{l}\right)  -I\left(  X_{1}:\cdots:X_{l}\right)  $ and optimize
it over all local measurements. Suppose that the state is classical to begin
with, meaning that it can be written as%
\begin{equation}
\sum_{x_{1},\ldots,x_{l}}p\left(  x_{1},\ldots,x_{l}\right)  \left\vert
x_{1}\right\rangle \left\langle x_{1}\right\vert _{A_{1}}\otimes\cdots
\otimes\left\vert x_{l}\right\rangle \left\langle x_{l}\right\vert _{A_{l}%
},\label{eq:multi-classical}%
\end{equation}
for some joint probability distribution $p\left(  x_{1},\ldots,x_{l}\right)  $
and orthonormal bases $\{\left\vert x_{i}\right\rangle _{A_{i}}\}$ for
$i\in\left\{  1,\ldots,l\right\}  $. Then there are local measurements that do
not change the state at all after they are performed, and the MSQ\ discord is
equal to zero. If the state cannot be written as above, then it cannot be
understood in a classical way, such that there does not exist a set of local
measurements that would leave the state undisturbed. In this sense, the
MSQ\ discord is a measure of multipartite quantum correlations between the
different parties and it is known that it is a faithful measure \cite{PHH08},
meaning that it is zero if and only if the state is multipartite classical as
written in (\ref{eq:multi-classical}). Other desirable properties for a
discord-like measure are described in \cite[Section~2.1]{CBRFA14}.

The CEMI is motivated by the concept of the monogamy of quantum entanglement,
that if two or more systems are highly entangled then any other systems cannot
be too entangled with them. On the other hand, states which are close to being
unentangled are highly shareable \cite{W89a}\ or extendible \cite{DPS05}, such
that there could be many other systems sharing the same correlations with
them. So to define the CEMI, we begin with a multipartite state on the systems
$A_{1}\cdots A_{l}$ and try to find a global state on these systems and some
others $A_{1}^{\prime}\cdots A_{l}^{\prime}$ that is consistent with the
original state, meaning that we recover the original state when tracing over
$A_{1}^{\prime}\cdots A_{l}^{\prime}$. From the aforementioned ideas, any
classical correlations can be shared with the extension systems $A_{1}%
^{\prime}\cdots A_{l}^{\prime}$ while entanglement cannot be shared. The
information gap $I\left(  A_{1}A_{1}^{\prime}:\cdots:A_{l}A_{l}^{\prime
}\right)  -I\left(  A_{1}^{\prime}:\cdots:A_{l}^{\prime}\right)  $ attempts to
subtract out the multipartite classical correlations that are shareable, so
that what is left is a measure of multipartite quantum entanglement. One then
optimizes this quantity by taking an infimum over all extension states. The
work of \cite{YHW08} fully justified this approach, proving that the CEMI\ is
a proper entanglement measure and bears many properties which are desirable
for such a measure. What was left open was to prove that the CEMI\ is a
faithful entanglement measure, meaning that it is equal to zero if and only if
the state on $A_{1}\cdots A_{l}$ is a fully separable (unentangled) state
\cite{W89}\ of the following form:%
\begin{equation}
\sum_{z}p\left(  z\right)  \sigma_{A_{1}}^{z}\otimes\cdots\otimes\sigma
_{A_{l}}^{z},
\end{equation}
where $p\left(  z\right)  $ is a probability distribution and $\sigma_{A_{i}%
}^{z}$ is a quantum state on system $A_{i}$.

\section{Summary of results}

The Fawzi-Renner inequality in (\ref{eq:FR}) has a number of implications for
entanglement theory and more general quantum correlations:\ it gives an
alternate method \cite{Winterconj,LW14}\ from \cite{BCY11}\ for establishing
the faithfulness of the squashed entanglement measure \cite{CW04}\ and it
allows for characterizing quantum states with discord \cite{Z00,OZ01}\ nearly
equal to zero as being approximate fixed points of entanglement breaking
channels \cite[Proposition~29]{SW14}.

The main objective of the present paper is to pursue extensions of these ideas
for multipartite quantum states and correlation measures. In particular, we
first demonstrate that the following \textquotedblleft local
recoverability\textquotedblright\ inequality is a consequence of the
inequality in (\ref{eq:FR}):%
\begin{multline}
I\left(  A_{1}A_{1}^{\prime}:\cdots:A_{l}A_{l}^{\prime}\right)  _{\rho
}-I\left(  A_{1}^{\prime}:\cdots:A_{l}^{\prime}\right)  _{\rho}%
\label{eq:local-recover}\\
\geq\left[  \frac{1}{2l}\left\Vert \rho_{A_{1}A_{1}^{\prime}\cdots A_{l}%
A_{l}^{\prime}}-\left(  \mathcal{R}_{A_{1}^{\prime}\rightarrow A_{1}%
A_{1}^{\prime}}^{1}\otimes\cdots\otimes\mathcal{R}_{A_{l}^{\prime}\rightarrow
A_{l}A_{l}^{\prime}}^{l}\right)  \left(  \rho_{A_{1}^{\prime}\cdots
A_{l}^{\prime}}\right)  \right\Vert _{1}\right]  ^{2},
\end{multline}
where $\mathcal{R}_{A_{1}^{\prime}\rightarrow A_{1}A_{1}^{\prime}}^{1}$,
\ldots, $\mathcal{R}_{A_{l}^{\prime}\rightarrow A_{l}A_{l}^{\prime}}^{l}$ are
local recovery maps. The implication of the above inequality is that if the
gap $I\left(  A_{1}A_{1}^{\prime}:\cdots:A_{l}A_{l}^{\prime}\right)  _{\rho
}-I\left(  A_{1}^{\prime}:\cdots:A_{l}^{\prime}\right)  _{\rho}$ is nearly
equal to zero, then the full state $\rho_{A_{1}A_{1}^{\prime}\cdots A_{l}%
A_{l}^{\prime}}$ is \textquotedblleft locally recoverable,\textquotedblright%
\ i.e., one can approximately recover it by performing the local recovery maps
$\mathcal{R}_{A_{1}^{\prime}\rightarrow A_{1}A_{1}^{\prime}}^{1}$, \ldots,
$\mathcal{R}_{A_{l}^{\prime}\rightarrow A_{l}A_{l}^{\prime}}^{l}$. The
converse of this statement is a direct consequence of the Alicki-Fannes
inequality \cite{AF04}, with a proof proceeding similarly to the steps in
(\ref{eq:discord-approx-1})-(\ref{eq:discord-approx-5}) and a dimension
dependence only on the systems $A_{1}$, \ldots, $A_{l}$. It might be possible
to improve upon the inequality in (\ref{eq:local-recover}), i.e., to have the
$l$-independent inequality:%
\begin{multline}
I\left(  A_{1}A_{1}^{\prime}:\cdots:A_{l}A_{l}^{\prime}\right)  _{\rho
}-I\left(  A_{1}^{\prime}:\cdots:A_{l}^{\prime}\right)  _{\rho}%
\label{eq:conjectured-ineq}\\
\geq-\log F\left(  \rho_{A_{1}A_{1}^{\prime}\cdots A_{l}A_{l}^{\prime}%
},\left(  \mathcal{R}_{A_{1}^{\prime}\rightarrow A_{1}A_{1}^{\prime}}%
^{1}\otimes\cdots\otimes\mathcal{R}_{A_{l}^{\prime}\rightarrow A_{l}%
A_{l}^{\prime}}^{l}\right)  \left(  \rho_{A_{1}^{\prime}\cdots A_{l}^{\prime}%
}\right)  \right)  .
\end{multline}
We elaborate more on this possibility in Section~\ref{sec:potential-improve}.

Regardless of whether the conjectured inequality in (\ref{eq:conjectured-ineq}%
) holds, we can already establish two consequences of the inequality in
(\ref{eq:local-recover}):

\begin{enumerate}
\item The multipartite symmetric quantum (MSQ)\ discord from \cite{PHH08} is
nearly equal to zero if and only if the multipartite state $\rho_{A_{1}\cdots
A_{l}}$ is locally recoverable after performing measurements on each of the
systems $A_{1}$, \ldots, $A_{l}$. Equivalently, such a state has MSQ discord
nearly equal to zero if and only if it is an approximate fixed point of a
tensor product of entanglement breaking channels. Recall that any entanglement
breaking channel can be written as a composition of a measurement channel
followed by a preparation channel \cite{HSR03}. We detail this result in
Section~\ref{sec:MSQ-faithful}.

\item The conditional entanglement of multipartite information (CEMI)\ from
\cite{YHW08} is faithful, i.e., it vanishes if and only if a multipartite
state $\rho_{A_{1}\cdots A_{l}}$ is fully separable. We detail this result in
Section~\ref{sec:CEMI-faithful}.
\end{enumerate}

Additional contributions of this paper are to show explicitly in
Section~\ref{sec:CEMI-upper-bound}\ that the CEMI\ is an upper bound on the
multipartite squashed entanglement from \cite{YHHHOS09,AHS08} and in
Section~\ref{sec:op-int}\ to give an operational interpretation of the
MSQ\ discord in terms of the partial state distribution protocol from
\cite{YHW08}. We conclude in Section~\ref{sec:conclusion}\ with a summary of
results and directions for future work.

\section{Local recoverability}

In this section, we give a proof of the local recoverability inequality in
(\ref{eq:local-recover}). We start with an explicit proof of the following
lemma, which is implicit in the partial state distribution protocol of
\cite{YHW08}:

\begin{lemma}
\label{lem:chain-for-gap}Let $\rho_{A_{1}A_{1}^{\prime}\cdots A_{l}%
A_{l}^{\prime}}$ be a multipartite quantum state. Then we have the following
identity:%
\begin{equation}
I\left(  A_{1}A_{1}^{\prime}:\cdots:A_{l}A_{l}^{\prime}\right)  _{\rho
}-I\left(  A_{1}^{\prime}:\cdots:A_{l}^{\prime}\right)  _{\rho}=\sum_{i=1}%
^{l}I\left(  A_{i};A_{1}^{i-1}A_{\left[  l\right]  \backslash\left\{
i\right\}  }^{\prime}|A_{i}^{\prime}\right)  ,
\end{equation}
where $A_{1}^{i-1}\equiv A_{i-1}\cdots A_{1}$ (interpreted to be empty if
$i=1$) and $A_{\left[  l\right]  \backslash\left\{  i\right\}  }^{\prime}$ is
a shorthand indicating all of the $A^{\prime}$ systems except for
$A_{i}^{\prime}$. In addition, the expansion on the right-hand side can
proceed in any order.
\end{lemma}

\begin{proof}
Consider that%
\begin{align}
&  I\left(  A_{1}A_{1}^{\prime}:\cdots:A_{l}A_{l}^{\prime}\right)  _{\rho
}-I\left(  A_{1}^{\prime}:\cdots:A_{l}^{\prime}\right)  _{\rho}\nonumber\\
&  =\sum_{i=1}^{l}H\left(  A_{i}A_{i}^{\prime}\right)  _{\rho}-H\left(
A_{1}A_{1}^{\prime}A_{2}A_{2}^{\prime}\cdots A_{l}A_{l}^{\prime}\right)
_{\rho}-\left[  \sum_{i=1}^{l}H\left(  A_{i}^{\prime}\right)  _{\rho}-H\left(
A_{1}^{\prime}A_{2}^{\prime}\cdots A_{l}^{\prime}\right)  _{\rho}\right] \\
&  =\sum_{i=1}^{l}H\left(  A_{i}|A_{i}^{\prime}\right)  _{\rho}-H\left(
A_{1}A_{2}\cdots A_{l}|A_{1}^{\prime}A_{2}^{\prime}\cdots A_{l}^{\prime
}\right)  _{\rho}\\
&  =\sum_{i=1}^{l}H\left(  A_{i}|A_{i}^{\prime}\right)  _{\rho}-\sum_{i=1}%
^{l}H\left(  A_{i}|A_{1}^{i-1}A_{1}^{\prime}A_{2}^{\prime}\cdots A_{l}%
^{\prime}\right)  _{\rho}\\
&  =\sum_{i=1}^{l}\left[  H\left(  A_{i}|A_{i}^{\prime}\right)  _{\rho
}-H\left(  A_{i}|A_{1}^{i-1}A_{1}^{\prime}A_{2}^{\prime}\cdots A_{l}^{\prime
}\right)  _{\rho}\right] \\
&  =\sum_{i=1}^{l}I\left(  A_{i};A_{1}^{i-1}A_{\left[  l\right]
\backslash\left\{  i\right\}  }^{\prime}|A_{i}^{\prime}\right)  .
\end{align}
The first equality is an expansion following from definitions. The second
equality uses the chain rule for conditional entropy, i.e., $H\left(
A|B\right)  =H\left(  AB\right)  -H\left(  B\right)  $. The third equality
follows from an inductive application of the chain rule for conditional
entropy. The final equality follows from an expansion for conditional mutual
information as $I\left(  A;B|C\right)  =H\left(  A|C\right)  -H\left(
A|CB\right)  $. The statement about expanding in an arbitrary order follows
because the expansion in the third equality can proceed in any order.
\end{proof}

\bigskip

\begin{proof}
[Proof of (\ref{eq:local-recover})]We can now easily prove the inequality in
(\ref{eq:local-recover}). From Lemma~\ref{lem:chain-for-gap}, we can conclude
that%
\begin{equation}
I\left(  A_{1}A_{1}^{\prime}:\cdots:A_{l}A_{l}^{\prime}\right)  _{\rho
}-I\left(  A_{1}^{\prime}:\cdots:A_{l}^{\prime}\right)  _{\rho}\geq I\left(
A_{i};A_{\left[  l\right]  \backslash\left\{  i\right\}  }A_{\left[  l\right]
\backslash\left\{  i\right\}  }^{\prime}|A_{i}^{\prime}\right)
\end{equation}
for all $i\in\left\{  1,2,\ldots,l\right\}  $ because a)\ the expansion there
can proceed in any order and b)\ the conditional mutual information is
non-negative \cite{PhysRevLett.30.434,LR73}. From the inequality in
(\ref{eq:FR}), we can then conclude that there exists a recovery map
$\mathcal{R}_{A_{i}^{\prime}\rightarrow A_{i}A_{i}^{\prime}}^{i}$ such that%
\begin{align}
&  I\left(  A_{1}A_{1}^{\prime}:\cdots:A_{l}A_{l}^{\prime}\right)  _{\rho
}-I\left(  A_{1}^{\prime}:\cdots:A_{l}^{\prime}\right)  _{\rho}\nonumber\\
&  \geq I\left(  A_{i};A_{\left[  l\right]  \backslash\left\{  i\right\}
}A_{\left[  l\right]  \backslash\left\{  i\right\}  }^{\prime}|A_{i}^{\prime
}\right) \\
&  \geq\frac{1}{4}\left\Vert \rho_{A_{1}A_{1}^{\prime}\cdots A_{l}%
A_{l}^{\prime}}-\mathcal{R}_{A_{i}^{\prime}\rightarrow A_{i}A_{i}^{\prime}%
}^{i}\left(  \rho_{A_{\left[  l\right]  \backslash\left\{  i\right\}  }%
A_{1}^{\prime}\cdots A_{l}^{\prime}}\right)  \right\Vert _{1}^{2},
\end{align}
which is equivalent to%
\begin{align}
&  2\sqrt{I\left(  A_{1}A_{1}^{\prime}:\cdots:A_{l}A_{l}^{\prime}\right)
_{\rho}-I\left(  A_{1}^{\prime}:\cdots:A_{l}^{\prime}\right)  _{\rho}%
}\nonumber\\
&  \geq\left\Vert \rho_{A_{1}A_{1}^{\prime}\cdots A_{l}A_{l}^{\prime}%
}-\mathcal{R}_{A_{i}^{\prime}\rightarrow A_{i}A_{i}^{\prime}}^{i}\left(
\rho_{A_{\left[  l\right]  \backslash\left\{  i\right\}  }A_{1}^{\prime}\cdots
A_{l}^{\prime}}\right)  \right\Vert _{1}\\
&  =\left\Vert \rho_{A_{1}A_{1}^{\prime}\cdots A_{l}A_{l}^{\prime}}-\left(
\mathcal{R}_{A_{i}^{\prime}\rightarrow A_{i}A_{i}^{\prime}}^{i}\circ
\text{Tr}_{A_{i}}\right)  \left(  \rho_{A_{1}A_{1}^{\prime}\cdots A_{l}%
A_{l}^{\prime}}\right)  \right\Vert _{1}.
\end{align}
Using the triangle inequality $l$ times and monotonicity of the trace distance
under quantum operations (i.e., that $\left\Vert \omega-\tau\right\Vert
_{1}\geq\left\Vert \mathcal{N}\left(  \omega\right)  -\mathcal{N}\left(
\tau\right)  \right\Vert _{1}$ for density operators $\omega$ and $\tau$ and a
quantum channel $\mathcal{N}$), we can then conclude that%
\begin{multline}
2l\sqrt{I\left(  A_{1}A_{1}^{\prime}:\cdots:A_{l}A_{l}^{\prime}\right)
_{\rho}-I\left(  A_{1}^{\prime}:\cdots:A_{l}^{\prime}\right)  _{\rho}}\\
\geq\left\Vert \rho_{A_{1}A_{1}^{\prime}\cdots A_{l}A_{l}^{\prime}}-\left(
\mathcal{R}_{A_{1}^{\prime}\rightarrow A_{1}A_{1}^{\prime}}^{1}\otimes
\cdots\otimes\mathcal{R}_{A_{l}^{\prime}\rightarrow A_{l}A_{l}^{\prime}}%
^{l}\right)  \left(  \rho_{A_{1}^{\prime}\cdots A_{l}^{\prime}}\right)
\right\Vert _{1},
\end{multline}
which is equivalent to (\ref{eq:local-recover}).
\end{proof}

\begin{remark}
\label{rem:extended-proof}The above proof demonstrates that there are in fact
$2^{l}$ inequalities that hold, depending on whether one chooses to apply the
trace-out-and-recovery maps or not. The inequality then takes on the following
form:%
\begin{multline}
I\left(  A_{1}A_{1}^{\prime}:\cdots:A_{l}A_{l}^{\prime}\right)  _{\rho
}-I\left(  A_{1}^{\prime}:\cdots:A_{l}^{\prime}\right)  _{\rho}\\
\geq\left[  \frac{1}{2\left\vert j^{l}\right\vert }\left\Vert \rho_{A_{1}%
A_{1}^{\prime}\cdots A_{l}A_{l}^{\prime}}-\left(  \left(  \mathcal{R}%
_{A_{1}^{\prime}\rightarrow A_{1}A_{1}^{\prime}}^{1}\circ\operatorname{Tr}%
_{A_{1}}\right)  ^{j_{1}}\otimes\cdots\otimes\left(  \mathcal{R}%
_{A_{l}^{\prime}\rightarrow A_{l}A_{l}^{\prime}}^{l}\circ\operatorname{Tr}%
_{A_{l}}\right)  ^{j_{l}}\right)  \left(  \rho_{A_{1}A_{1}^{\prime}\cdots
A_{l}A_{l}^{\prime}}\right)  \right\Vert _{1}\right]  ^{2},
\end{multline}
where $j^{l}\equiv j_{1}\cdots j_{l}$ is a binary string indicating which
recovery maps are applied and $\left\vert j^{l}\right\vert $ is the number of
ones in $j^{l}$ if $j^{l}$ is not the all-zeros bit string, with $\left\vert
j^{l}\right\vert $ otherwise being equal to one.
\end{remark}

\section{Approximate faithfulness of the MSQ discord}

\label{sec:MSQ-faithful}In this section, we provide a generalization of the
approximate faithfulness of quantum discord \cite[Proposition~29]{SW14}\ to
the multipartite case. In particular, recall the multipartite symmetric
quantum (MSQ)\ discord from \cite{PHH08}:%
\begin{equation}
D\left(  \overline{A_{1}}:\cdots:\overline{A_{l}}\right)  _{\rho}\equiv
I\left(  A_{1}:\cdots:A_{l}\right)  _{\rho}-\sup_{\left\{  \mathcal{M}%
_{A_{1}\rightarrow X_{1}}^{1},\ldots,\mathcal{M}_{A_{l}\rightarrow X_{l}}%
^{l}\right\}  }I\left(  X_{1}:\cdots:X_{l}\right)  _{\omega},
\end{equation}
where $\rho_{A_{1}\cdots A_{l}}$ is a multipartite quantum state and
$\omega_{X_{1}\cdots X_{l}}$ is the state resulting from local measurements of
$\rho_{A_{1}\cdots A_{l}}$ according to the measurement maps $\mathcal{M}%
_{A_{1}\rightarrow X_{1}}^{1},\ldots,\mathcal{M}_{A_{l}\rightarrow X_{l}}^{l}%
$:%
\begin{equation}
\omega_{X_{1}\cdots X_{l}}\equiv\left(  \mathcal{M}_{A_{1}\rightarrow X_{1}%
}^{1}\otimes\cdots\otimes\mathcal{M}_{A_{l}\rightarrow X_{l}}^{l}\right)
\left(  \rho_{A_{1}\cdots A_{l}}\right)  .\label{eq:omega-state}%
\end{equation}
The measurement map $\mathcal{M}_{A_{i}\rightarrow X_{i}}^{i}$ is defined as%
\begin{equation}
\mathcal{M}_{A_{i}\rightarrow X_{i}}^{i}\left(  \sigma_{A_{i}}\right)
\equiv\sum_{x}\text{Tr}\left\{  \Lambda_{A_{i}}^{x}\sigma_{A_{i}}\right\}
\left\vert x\right\rangle \left\langle x\right\vert _{A_{i}}%
\end{equation}
for some positive semidefinite operators $\Lambda_{A_{i}}^{x}$ which sum to
the identity and where $\left\{  \left\vert x\right\rangle _{A_{i}}\right\}  $
is an orthonormal basis for the system $A_{i}$.

\begin{proposition}
[Approximate faithfulness]The MSQ\ discord is nearly equal to zero if and only
if $\rho_{A_{1}\cdots A_{l}}$ is an approximate fixed point of a tensor
product of entanglement breaking (EB)\ channels $\mathcal{E}_{A_{1}}^{1}$,
\ldots, $\mathcal{E}_{A_{l}}^{l}$. That is, suppose that there exist
EB\ channels $\mathcal{E}_{A_{1}}^{1}$, \ldots, $\mathcal{E}_{A_{l}}^{l}$ such
that%
\begin{equation}
\left\Vert \rho_{A_{1}\cdots A_{l}}-\left(  \mathcal{E}_{A_{1}}^{1}%
\otimes\cdots\otimes\mathcal{E}_{A_{l}}^{l}\right)  \left(  \rho_{A_{1}\cdots
A_{l}}\right)  \right\Vert _{1}\leq\varepsilon
\end{equation}
for some $\varepsilon\in\left[  0,1\right]  $. Then%
\begin{equation}
D\left(  \overline{A_{1}}:\cdots:\overline{A_{l}}\right)  _{\rho}\leq\left(
l+1\right)  h_{2}\left(  \varepsilon/2\right)  +\varepsilon\sum_{i=1}^{l}%
\log\left(  \left\vert A_{i}\right\vert \right)  , \label{eq:small-discord}%
\end{equation}
where $h_{2}\left(  \varepsilon\right)  $ is the binary entropy with the
property that $\lim_{\varepsilon\searrow0}h_{2}\left(  \varepsilon\right)
=0$. Conversely, suppose that%
\begin{equation}
D\left(  \overline{A_{1}}:\cdots:\overline{A_{l}}\right)  _{\rho}%
\leq\varepsilon
\end{equation}
for some $\varepsilon>0$. Then there exist EB\ channels $\mathcal{E}_{A_{1}%
}^{1}$, \ldots, $\mathcal{E}_{A_{l}}^{l}$ such that%
\begin{equation}
\left\Vert \rho_{A_{1}\cdots A_{l}}-\left(  \mathcal{E}_{A_{1}}^{1}%
\otimes\cdots\otimes\mathcal{E}_{A_{l}}^{l}\right)  \left(  \rho_{A_{1}\cdots
A_{l}}\right)  \right\Vert _{1}\leq2l\sqrt{\varepsilon}.
\label{eq:small-discord-other}%
\end{equation}

\end{proposition}

\begin{proof}
The proof of the inequality in (\ref{eq:small-discord}) proceeds exactly as in
the proof of \cite[Proposition 29]{SW14}. Consider that every EB\ channel can
be written as a composition of a measurement map and a preparation
\cite{HSR03}, i.e., $\mathcal{E}_{A_{i}}^{i}=\mathcal{P}_{X_{i}\rightarrow
A_{i}}^{i}\circ\mathcal{M}_{A_{i}\rightarrow X_{i}}^{i}$. Then%
\begin{align}
D\left(  \overline{A_{1}}:\cdots:\overline{A_{l}}\right)  _{\rho} &  =I\left(
A_{1}:\cdots:A_{l}\right)  _{\rho}-\sup_{\left\{  \mathcal{M}_{A_{1}%
\rightarrow X_{1}}^{1},\ldots,\mathcal{M}_{A_{l}\rightarrow X_{l}}%
^{l}\right\}  }I\left(  X_{1}:\cdots:X_{l}\right)  _{\omega}%
\label{eq:discord-approx-1}\\
&  \leq I\left(  A_{1}:\cdots:A_{l}\right)  _{\rho}-I\left(  X_{1}%
:\cdots:X_{l}\right)  _{\bigotimes\limits_{i}\mathcal{M}_{A_{i}\rightarrow
X_{i}}^{i}\left(  \rho\right)  }\\
&  \leq I\left(  A_{1}:\cdots:A_{l}\right)  _{\rho}-I\left(  A_{1}%
:\cdots:A_{l}\right)  _{\bigotimes\limits_{i}\mathcal{P}_{X_{i}\rightarrow
A_{i}}^{i}\circ\mathcal{M}_{A_{i}\rightarrow X_{i}}^{i}\left(  \rho\right)
}\\
&  =I\left(  A_{1}:\cdots:A_{l}\right)  _{\rho}-I\left(  A_{1}:\cdots
:A_{l}\right)  _{\bigotimes\limits_{i}\mathcal{E}_{A_{i}}^{i}\left(
\rho\right)  }\\
&  \leq\left(  l+1\right)  h_{2}\left(  \varepsilon/2\right)  +\varepsilon
\sum_{i=1}^{l}\log\left(  \left\vert A_{i}\right\vert \right)
.\label{eq:discord-approx-5}%
\end{align}
The first inequality follows by choosing the measurement maps not to be the
optimal ones, but instead the ones making up the first part of the
EB\ channels $\left\{  \mathcal{E}_{A_{i}}^{i}\right\}  $. The second
inequality follows from the fact that the multipartite information is monotone
under local operations (here being the processing of the measured systems
according to the preparation maps). The last inequality is a consequence of
the Fannes-Audenaert inequality \cite{A07}, which states that%
\begin{equation}
\left\vert H\left(  \rho\right)  -H\left(  \sigma\right)  \right\vert \leq
T\log\left(  d-1\right)  +h_{2}\left(  T\right)  ,
\end{equation}
with $T=\frac{1}{2}\left\Vert \rho-\sigma\right\Vert _{1}$ and $d$ the
dimension of the density operators $\rho$ and $\sigma$.

After recalling that any quantum channel (incuding measurement maps) can be
understood as an isometric embedding of the input in a tensor-product Hilbert
space followed by a partial trace \cite{S55}, we can see that
(\ref{eq:small-discord-other}) is a consequence of the inequality in
(\ref{eq:local-recover}). Specifically, for a particular set of measurements,
we can write%
\begin{equation}
I\left(  A_{1}:\cdots:A_{l}\right)  _{\rho}-I\left(  X_{1}:\cdots
:X_{l}\right)  _{\omega}=I\left(  X_{1}E_{1}:\cdots:X_{l}E_{l}\right)
_{\omega}-I\left(  X_{1}:\cdots:X_{l}\right)  _{\omega}%
,\label{eq:MI-rewrite-1}%
\end{equation}
where%
\begin{equation}
\omega_{X_{1}E_{1}\cdots X_{l}E_{l}}\equiv\left(  \mathcal{U}_{A_{1}%
\rightarrow X_{1}E_{1}}^{\mathcal{M}^{1}}\otimes\cdots\otimes\mathcal{U}%
_{A_{l}\rightarrow X_{l}E_{l}}^{\mathcal{M}^{l}}\right)  \left(  \rho
_{A_{1}\cdots A_{l}}\right)  \label{eq:full-state}%
\end{equation}
and $\mathcal{U}_{A_{i}\rightarrow X_{i}E_{i}}^{\mathcal{M}^{i}}$ is an
isometric CPTP map, so that%
\begin{equation}
\mathcal{U}_{A_{i}\rightarrow X_{i}E_{i}}^{\mathcal{M}^{i}}\left(
\cdot\right)  \equiv U_{A_{i}\rightarrow X_{i}E_{i}}^{\mathcal{M}^{i}}\left(
\cdot\right)  \left[  U_{A_{i}\rightarrow X_{i}E_{i}}^{\mathcal{M}^{i}%
}\right]  ^{\dag},
\end{equation}
where $U_{A_{i}\rightarrow X_{i}E_{i}}^{\mathcal{M}^{i}}$ is an isometric
extension of the measurement map $\mathcal{M}_{A_{i}\rightarrow X_{i}}^{i}$.
Then (\ref{eq:MI-rewrite-1}) follows because the multipartite information is
invariant under local isometries, as one can see from its definition in
(\ref{eq:multi-info}) and invariance of quantum entropy under isometries. The
inequality (\ref{eq:small-discord-other}) then follows because there exist
recovery maps $\mathcal{R}_{X_{1}\rightarrow X_{1}E_{1}}^{1}$, \ldots,
$\mathcal{R}_{X_{l}\rightarrow X_{l}E_{l}}^{l}$ such that%
\begin{align}
&  I\left(  X_{1}E_{1}:\cdots:X_{l}E_{l}\right)  _{\omega}-I\left(
X_{1}:\cdots:X_{l}\right)  _{\omega}\nonumber\\
&  \geq\left[  \frac{1}{2l}\left\Vert \omega_{X_{1}E_{1}\cdots X_{l}E_{l}%
}-\left(  \mathcal{R}_{X_{1}\rightarrow X_{1}E_{1}}^{1}\otimes\cdots
\otimes\mathcal{R}_{X_{l}\rightarrow X_{l}E_{l}}^{l}\right)  \left(
\omega_{X_{1}\cdots X_{l}}\right)  \right\Vert _{1}\right]  ^{2}\\
&  \geq\left[  \frac{1}{2l}\left\Vert \rho_{A_{1}\cdots A_{l}}-\left(
\mathcal{P}_{X_{1}\rightarrow A_{1}}^{1}\otimes\cdots\otimes\mathcal{P}%
_{X_{l}\rightarrow A_{l}}^{l}\right)  \left(  \omega_{X_{1}\cdots X_{l}%
}\right)  \right\Vert _{1}\right]  ^{2}\\
&  =\left[  \frac{1}{2l}\left\Vert \rho_{A_{1}\cdots A_{l}}-\left(
\mathcal{E}_{A_{1}}^{1}\otimes\cdots\otimes\mathcal{E}_{A_{l}}^{l}\right)
\left(  \rho_{A_{1}\cdots A_{l}}\right)  \right\Vert _{1}\right]  ^{2}.
\end{align}
The first inequality is a consequence of (\ref{eq:local-recover}). We define
the following CPTP\ maps:%
\begin{multline}
\mathcal{T}_{X_{i}E_{i}\rightarrow A_{i}}^{i}\left(  \gamma_{X_{i}E_{i}%
}\right)  \equiv\left[  U_{A_{i}\rightarrow X_{i}E_{i}}^{\mathcal{M}^{i}%
}\right]  ^{\dag}\gamma_{X_{i}E_{i}}U_{A_{i}\rightarrow X_{i}E_{i}%
}^{\mathcal{M}^{i}}\\
+\text{Tr}\left\{  \left(  I_{X_{i}E_{i}}-U_{A_{i}\rightarrow X_{i}E_{i}%
}^{\mathcal{M}^{i}}\left[  U_{A_{i}\rightarrow X_{i}E_{i}}^{\mathcal{M}^{i}%
}\right]  ^{\dag}\right)  \gamma_{X_{i}E_{i}}\right\}  \sigma_{A_{i}}^{i},
\end{multline}
where $\sigma_{A_{i}}^{i}$ is some state on system $A_{i}$. Observe that%
\begin{equation}
\left(  \mathcal{T}_{X_{1}E_{1}\rightarrow A_{1}}^{1}\otimes\cdots
\otimes\mathcal{T}_{X_{l}E_{l}\rightarrow A_{l}}^{l}\right)  \left(
\omega_{X_{1}E_{1}\cdots X_{l}E_{l}}\right)  =\rho_{A_{1}\cdots A_{l}}.
\end{equation}
Then the second inequality above follows by defining the preparation maps
$\mathcal{P}_{X_{i}\rightarrow A_{i}}^{i}$ as%
\begin{equation}
\mathcal{P}_{X_{i}\rightarrow A_{i}}^{i}\equiv\mathcal{T}_{X_{i}%
E_{i}\rightarrow A_{i}}^{i}\circ\mathcal{R}_{X_{i}\rightarrow X_{i}E_{i}}^{i},
\end{equation}
and noting that the trace distance does not increase under the CPTP\ map
$\mathcal{T}_{X_{1}E_{1}\rightarrow A_{1}}^{1}\otimes\cdots\otimes
\mathcal{T}_{X_{l}E_{l}\rightarrow A_{l}}^{l}$. (The maps $\mathcal{P}%
_{X_{i}\rightarrow A_{i}}^{i}$ are preparations because they act on classical
registers.)\ The last equality follows from the definition of $\omega
_{X_{1}\cdots X_{l}}$ in (\ref{eq:omega-state}) and the fact that any
composition of a measurement map followed by a preparation map is entanglement
breaking \cite{HSR03}.
\end{proof}

\section{Faithfulness of the CEMI}

\label{sec:CEMI-faithful}The conditional entanglement of multipartite
information (CEMI)\ is an entanglement measure defined in \cite{YHW08}. It
bears some similarities with the squashed entanglement \cite{CW04} and its
multipartite version \cite{YHHHOS09,AHS08}. In \cite{YHW07,YHW08}, the CEMI
was shown to be non-negative, monotone under local operations and classical
communication, convex, additive, asymptotically continuous, and equal to zero
for separable states. It is not known to be monogamous. Given a multipartite
state $\rho_{A_{1}\cdots A_{l}}$, the CEMI is defined as follows:%
\begin{equation}
E_{I}\left(  A_{1}:\cdots:A_{l}\right)  _{\rho}\equiv\frac{1}{2}\inf
_{\rho_{A_{1}A_{1}^{\prime}\cdots A_{l}A_{l}^{\prime}}}I\left(  A_{1}%
A_{1}^{\prime}:\cdots:A_{l}A_{l}^{\prime}\right)  _{\rho}-I\left(
A_{1}^{\prime}:\cdots:A_{l}^{\prime}\right)  _{\rho}, \label{eq:CEMI}%
\end{equation}
where the infimum is over all extensions $\rho_{A_{1}A_{1}^{\prime}\cdots
A_{l}A_{l}^{\prime}}$ of $\rho_{A_{1}\cdots A_{l}}$, i.e.,%
\begin{equation}
\rho_{A_{1}\cdots A_{l}}=\text{Tr}_{A_{1}^{\prime}\cdots A_{l}^{\prime}%
}\left\{  \rho_{A_{1}A_{1}^{\prime}\cdots A_{l}A_{l}^{\prime}}\right\}  .
\end{equation}

In this section, we prove that the CEMI is faithful, i.e., equal to zero if
and only if the state $\rho_{A_{1}\cdots A_{l}}$ is separable. Before doing
so, it may be helpful to review the if-part of this theorem from \cite{YHW08}.
If $\rho_{A_{1}\cdots A_{l}}$ is separable, then it has a decomposition of the
following form \cite{W89}:%
\begin{equation}
\rho_{A_{1}\cdots A_{l}}\equiv\sum_{x}p_{X}\left(  x\right)  \sigma_{A_{1}%
}^{1,x}\otimes\cdots\otimes\sigma_{A_{l}}^{l,x},
\end{equation}
for a probability distribution $p_{X}$ and states $\left\{  \sigma_{A_{1}%
}^{1,x}\right\}  $, \ldots, $\left\{  \sigma_{A_{l}}^{l,x}\right\}  $. In this
case, one particular extension of this state has the following form:%
\begin{equation}
\sum_{x}p_{X}\left(  x\right)  \sigma_{A_{1}}^{1,x}\otimes\left\vert
x\right\rangle \left\langle x\right\vert _{A_{1}^{\prime}}\otimes\cdots
\otimes\sigma_{A_{l}}^{l,x}\otimes\left\vert x\right\rangle \left\langle
x\right\vert _{A_{l}^{\prime}}.
\end{equation}
It is then clear for this particular extension that%
\begin{equation}
I\left(  A_{1}^{\prime}:\cdots:A_{l}^{\prime}\right)  \geq I\left(  A_{1}%
A_{1}^{\prime}:\cdots:A_{l}A_{l}^{\prime}\right)  ,
\end{equation}
because one can produce the systems $A_{1}$, \ldots, $A_{l}$ by local
preparation maps of the form:%
\begin{equation}
\left(  \cdot\right)  \rightarrow\sum_{x}\left\vert x\right\rangle
\left\langle x\right\vert _{A_{i}^{\prime}}\left(  \cdot\right)  \left\vert
x\right\rangle \left\langle x\right\vert _{A_{i}^{\prime}}\otimes\sigma
_{A_{i}}^{i,x}.
\end{equation}
Combined with the inequality in (\ref{eq:multi-monotone}) and the definition
of $E_{I}$ in (\ref{eq:CEMI}), we find that $E_{I}$ is equal to zero if the
state is separable.

We now establish the only-if-part of faithfulness of CEMI, which is a
consequence of the following proposition:

\begin{proposition}
The CEMI of a multipartite state $\rho_{A_{1}\cdots A_{l}}$ obeys the
following bound:%
\begin{equation}
E_{I}\left(  A_{1}:\cdots:A_{l}\right)  _{\rho}\geq\frac{1}{16\cdot\left(
l+1\right)  ^{4}}\left(  \sum_{i=2}^{l}\left\vert A_{i}\right\vert
^{2}\right)  ^{-2}\left\Vert \rho_{A_{1}\cdots A_{l}}-\operatorname{SEP}%
\left(  A_{1}:\cdots:A_{l}\right)  \right\Vert _{1}^{4},
\end{equation}
where $\left\Vert \rho_{A_{1}\cdots A_{l}}-\operatorname{SEP}\left(
A_{1}:\cdots:A_{l}\right)  \right\Vert _{1}$ is the trace distance from
$\rho_{A_{1}\cdots A_{l}}$ to the set of multipartite separable states.
\end{proposition}

\begin{proof}
The proof of this proposition proceeds along the lines outlined in
\cite{Winterconj,LW14}, an analysis which is repeated in both \cite{SW14}\ and
\cite{FR14}. Let $\varepsilon_{\rho}$ denote the value of the following
quantity for a particular extension $\rho_{A_{1}A_{1}^{\prime}\cdots
A_{l}A_{l}^{\prime}}$:%
\begin{equation}
\varepsilon_{\rho}=I\left(  A_{1}A_{1}^{\prime}:\cdots:A_{l}A_{l}^{\prime
}\right)  _{\rho}-I\left(  A_{1}^{\prime}:\cdots:A_{l}^{\prime}\right)
_{\rho}.
\end{equation}
From Remark~\ref{rem:extended-proof}, we know that there exist recovery maps
$\mathcal{R}_{A_{1}^{\prime}\rightarrow A_{1}A_{1}^{\prime}}^{1}$, \ldots,
$\mathcal{R}_{A_{l}^{\prime}\rightarrow A_{l}A_{l}^{\prime}}^{l}$ such that
the following inequalities hold%
\begin{equation}
\varepsilon_{\rho}\geq\left[  \frac{1}{2l}\left\Vert \rho_{A_{1}A_{1}^{\prime
}\cdots A_{l}A_{l}^{\prime}}-\left(  \left(  \mathcal{R}_{A_{1}^{\prime
}\rightarrow A_{1}A_{1}^{\prime}}^{1}\circ\operatorname{Tr}_{A_{1}}\right)
^{j_{1}}\otimes\cdots\otimes\left(  \mathcal{R}_{A_{l}^{\prime}\rightarrow
A_{l}A_{l}^{\prime}}^{l}\circ\operatorname{Tr}_{A_{l}}\right)  ^{j_{l}%
}\right)  \left(  \rho_{A_{1}A_{1}^{\prime}\cdots A_{l}A_{l}^{\prime}}\right)
\right\Vert _{1}\right]  ^{2},
\end{equation}
where $j^{l}\equiv j_{1}\cdots j_{l}$ is a binary string indicating which
recovery maps are applied. Setting%
\begin{equation}
\delta_{\rho}\equiv2\sqrt{\varepsilon_{\rho}},
\end{equation}
these inequalities are then equivalent to the following ones:%
\begin{equation}
l\cdot\delta_{\rho}\geq\left\Vert \rho_{A_{1}A_{1}^{\prime}\cdots A_{l}%
A_{l}^{\prime}}-\left(  \left(  \mathcal{R}_{A_{1}^{\prime}\rightarrow
A_{1}A_{1}^{\prime}}^{1}\circ\operatorname{Tr}_{A_{1}}\right)  ^{j_{1}}%
\otimes\cdots\otimes\left(  \mathcal{R}_{A_{l}^{\prime}\rightarrow A_{l}%
A_{l}^{\prime}}^{l}\circ\operatorname{Tr}_{A_{l}}\right)  ^{j_{l}}\right)
\left(  \rho_{A_{1}A_{1}^{\prime}\cdots A_{l}A_{l}^{\prime}}\right)
\right\Vert _{1}.\label{eq:starting-ineq-mult}%
\end{equation}
Let $A_{j}^{k}\equiv A_{j,1}\cdots A_{j,k}$ for $j\in\left\{  1,\ldots
,l\right\}  $, and let $\Omega_{A_{1}^{k}\cdots A_{l}^{k}}$ denote the
following state, which results from many repeated atttempts at local recovery:%
\begin{multline}
\Omega_{A_{1}^{k}\cdots A_{l}^{k}A_{1}^{\prime}\cdots A_{l}^{\prime}}\equiv\\
\left(  \left(  \mathcal{R}_{A_{1}^{\prime}\rightarrow A_{1,k}A_{1}^{\prime}%
}^{1}\circ\cdots\circ\mathcal{R}_{A_{1}^{\prime}\rightarrow A_{1,1}%
A_{1}^{\prime}}^{1}\right)  \otimes\cdots\otimes\left(  \mathcal{R}%
_{A_{l}^{\prime}\rightarrow A_{l,k}A_{l}^{\prime}}^{l}\circ\cdots
\circ\mathcal{R}_{A_{l}^{\prime}\rightarrow A_{l,1}A_{l}^{\prime}}^{l}\right)
\right)  \left(  \rho_{A_{1}A_{1}^{\prime}\cdots A_{l}A_{l}^{\prime}}\right)
.
\end{multline}
From (\ref{eq:starting-ineq-mult}), the triangle inequality, monotonicity of
the trace distance under quantum operations, and the fact that the recovery
maps $\mathcal{R}_{A_{1}^{\prime}\rightarrow A_{1}A_{1}^{\prime}}^{1}$,
\ldots, $\mathcal{R}_{A_{l}^{\prime}\rightarrow A_{l}A_{l}^{\prime}}^{l}$
commute with each other because they act on different systems, we can conclude
that all of the following inequalities hold%
\begin{equation}
\left\Vert \rho_{A_{1}\cdots A_{l}}-\Omega_{A_{1,x_{1}}\cdots A_{l,x_{l}}%
}\right\Vert _{1}\leq lk\cdot\delta_{\rho}\label{eq:k-ext-ineqs}%
\end{equation}
for all tuples $\left(  x_{1},\ldots,x_{l}\right)  $ where $x_{i}\in\left\{
1,\ldots,k\right\}  $ and $i\in\left\{  1,\ldots,l\right\}  $. We can then
symmetrize the systems $A_{i}^{k}$ according to the random permutation:%
\begin{equation}
\overline{\Pi}_{A_{i}^{k}}\left(  \cdot\right)  \equiv\frac{1}{k!}\sum_{\pi\in
S_{k}}W_{A_{i,1}\cdots A_{i,k}}^{\pi}\left(  \cdot\right)  \left(
W_{A_{i,1}\cdots A_{i,k}}^{\pi}\right)  ^{\dag},\label{eq:randomizing-channel}%
\end{equation}
where $W_{A_{i,1}\cdots A_{i,k}}^{\pi}$ is a unitary representation of the
permutation $\pi$ which acts on the $k$-partite space $\mathcal{H}_{A_{i,1}%
}\otimes\cdots\otimes\mathcal{H}_{A_{i,k}}$\ as%
\begin{equation}
W_{A_{i,1}\cdots A_{i,k}}^{\pi}\left\vert m_{1}\right\rangle _{A_{i,1}}%
\otimes\cdots\otimes\left\vert m_{k}\right\rangle _{A_{i,k}}=\left\vert
m_{\pi^{-1}(1)}\right\rangle _{A_{i,1}}\otimes\cdots\otimes\left\vert
m_{\pi^{-1}(k)}\right\rangle _{A_{i,k}}.
\end{equation}
This leads to the multipartite extension state:%
\begin{equation}
\overline{\Omega}_{A_{1}^{k}\cdots A_{l}^{k}}\equiv\left(  \overline{\Pi
}_{A_{1}^{k}}\otimes\cdots\otimes\overline{\Pi}_{A_{l}^{k}}\otimes
\text{Tr}_{A_{1}^{\prime}\cdots A_{l}^{\prime}}\right)  \left(  \Omega
_{A_{1}^{k}\cdots A_{l}^{k}A_{1}^{\prime}\cdots A_{l}^{\prime}}\right)  .
\end{equation}
Combining convexity of the trace norm with the inequalities in
(\ref{eq:k-ext-ineqs}) gives the following inequality:%
\begin{equation}
\left\Vert \rho_{A_{1}\cdots A_{l}}-\overline{\Omega}_{A_{1,1}A_{2,1}\cdots
A_{l,1}}\right\Vert _{1}\leq lk\cdot\delta_{\rho},\label{eq:1st-one}%
\end{equation}
quantifying the distance between $\rho_{A_{1}\cdots A_{l}}$\ and the set of
multipartite $k$-extendible states \cite{W89a,DPS05}. By applying
Proposition~\ref{prop:multi-de-finetti} from the appendix, we know that%
\begin{equation}
\left\Vert \overline{\Omega}_{A_{1,1}A_{2,1}\cdots A_{l,1}}-\text{SEP}\left(
A_{1}:\cdots:A_{l}\right)  \right\Vert _{1}\leq\frac{2}{k}\left(  \sum
_{i=2}^{l}\left\vert A_{i}\right\vert ^{2}\right)  .\label{eq:2nd-one}%
\end{equation}
By choosing%
\begin{equation}
k=\left\lfloor \sqrt{\frac{2}{\delta_{\rho}}}\left(  \sum_{i=2}^{l}\left\vert
A_{i}\right\vert ^{2}\right)  ^{1/2}\right\rfloor ,
\end{equation}
and combining (\ref{eq:1st-one}) and (\ref{eq:2nd-one}) with the triangle
inequality, we find that%
\begin{align}
\left\Vert \rho_{A_{1}\cdots A_{l}}-\text{SEP}\left(  A_{1}:\cdots
:A_{l}\right)  \right\Vert _{1}  & \leq\left(  l+1\right)  \left(  \sum
_{i=2}^{l}\left\vert A_{i}\right\vert ^{2}\right)  ^{1/2}\sqrt{2\delta_{\rho}%
}\\
& =2\left(  l+1\right)  \left(  \sum_{i=2}^{l}\left\vert A_{i}\right\vert
^{2}\right)  ^{1/2}\sqrt[4]{\varepsilon_{\rho}}.
\end{align}
Since the inequality holds independently of the particular extension
$\rho_{A_{1}A_{1}^{\prime}\cdots A_{l}A_{l}^{\prime}}$, we can rearrange it
and take an infimum over all such extensions to find that%
\begin{equation}
E_{I}\left(  A_{1}:\cdots:A_{l}\right)  _{\rho}\geq\frac{1}{16\cdot\left(
l+1\right)  ^{4}}\left(  \sum_{i=2}^{l}\left\vert A_{i}\right\vert
^{2}\right)  ^{-2}\left\Vert \rho_{A_{1}\cdots A_{l}}-\text{SEP}\left(
A_{1}:\cdots:A_{l}\right)  \right\Vert _{1}^{4}.
\end{equation}

\end{proof}

\begin{remark}
The above approach follows that given by Li and Winter in
\cite{Winterconj,LW14}. The appendix of \cite{LW14} sketches an approach for
the multipartite squashed entanglement (the definition of which is recalled in
the next section)\ but remarked that there were difficulties in completing the
proof because in this case the local recovery map acts on the same extension
system and it is not clear whether inequalities like those in
\eqref{eq:k-ext-ineqs} would hold. This difficulty is removed in our setting
here (for the CEMI) because the local recovery maps act on different
subsystems of the extension system. It still remains an open question to
establish faithfulness of the multipartite squashed entanglement.
\end{remark}

\subsection{CEMI\ is an upper bound on multipartite squashed entanglement}

\label{sec:CEMI-upper-bound}The conditional multipartite information of
$\sigma_{A_{1}\cdots A_{l}E}$ is defined as%
\begin{equation}
I\left(  A_{1}:\cdots:A_{l}|E\right)  _{\sigma}\equiv H\left(  A_{1}|E\right)
_{\sigma}+\cdots H\left(  A_{l}|E\right)  _{\sigma}-H\left(  A_{1}\cdots
A_{l}|E\right)  _{\sigma}.
\end{equation}
From this, one can define the multipartite squashed entanglement of a state
$\rho_{A_{1}\cdots A_{l}}$\ as \cite{YHHHOS09,AHS08}:%
\begin{equation}
E_{\text{sq}}\left(  A_{1}:\cdots:A_{l}\right)  _{\rho}\equiv\frac{1}{2}%
\inf_{\rho_{A_{1}\cdots A_{l}E}}I\left(  A_{1}:\cdots:A_{l}|E\right)  _{\rho},
\end{equation}
where the infimum is over all extensions $\rho_{A_{1}\cdots A_{l}E}$ of
$\rho_{A_{1}\cdots A_{l}}$. The following proposition generalizes
Proposition~3 of \cite{YHW07}\ to the multipartite setting but is however
implicit in their concluding statement \textquotedblleft All conclusions for
the bipartite case can be similarly deduced\textquotedblright. (Nevertheless,
it seems worthwhile to produce a short explicit proof.)

\begin{proposition}
The multipartite squashed entanglement $E_{\operatorname{sq}}\left(
A_{1}:\cdots:A_{l}\right)  _{\rho}$ is never larger than the CEMI
$E_{I}\left(  A_{1}:\cdots:A_{l}\right)  _{\rho}$:%
\begin{equation}
E_{\operatorname{sq}}\left(  A_{1}:\cdots:A_{l}\right)  _{\rho}\leq
E_{I}\left(  A_{1}:\cdots:A_{l}\right)  _{\rho}. \label{eq:CEMI-squash}%
\end{equation}

\end{proposition}

\begin{proof}
Consider that \cite{YHHHOS09}%
\begin{equation}
I\left(  A_{1}:\cdots:A_{l}|E\right)  _{\rho}=\sum_{i=1}^{l}I\left(
A_{i};A_{1}^{i-1}|E\right)  .
\end{equation}
While, from Lemma~\ref{lem:chain-for-gap}, an additional application of the
chain rule, and strong subadditivity, we have that%
\begin{align}
I\left(  A_{1}A_{1}^{\prime}:\cdots:A_{l}A_{l}^{\prime}\right)  _{\rho
}-I\left(  A_{1}^{\prime}:\cdots:A_{l}^{\prime}\right)  _{\rho}  &
=\sum_{i=1}^{l}I\left(  A_{i};A_{1}^{i-1}A_{\left[  l\right]  \backslash
\left\{  i\right\}  }^{\prime}\middle|A_{i}^{\prime}\right) \\
&  =\sum_{i=1}^{l}I\left(  A_{i};A_{1}^{i-1}\middle|A_{\left[  l\right]
\backslash\left\{  i\right\}  }^{\prime}A_{i}^{\prime}\right)  +I\left(
A_{i};A_{\left[  l\right]  \backslash\left\{  i\right\}  }^{\prime
}\middle|A_{i}^{\prime}\right) \\
&  =\sum_{i=1}^{l}I\left(  A_{i};A_{1}^{i-1}|A_{1}^{\prime}\cdots
A_{l}^{\prime}\right)  +I\left(  A_{i};A_{\left[  l\right]  \backslash\left\{
i\right\}  }^{\prime}\middle|A_{i}^{\prime}\right) \\
&  \geq\sum_{i=1}^{l}I\left(  A_{i};A_{1}^{i-1}|A_{1}^{\prime}\cdots
A_{l}^{\prime}\right) \\
&  =I\left(  A_{1}:\cdots:A_{l}|A_{1}^{\prime}\cdots A_{l}^{\prime}\right)
_{\rho}\\
&  \geq E_{\operatorname{sq}}\left(  A_{1}:\cdots:A_{l}\right)  _{\rho}.
\end{align}
Since the above chain holds independently of the particular extension, this
establishes (\ref{eq:CEMI-squash}).
\end{proof}

\section{Partial state distribution and operational interpretations}

\label{sec:PSD}In this section, we review the partial state distribution
protocol from \cite{YHW08}\ and discuss how it gives an operational
interpretation for the MSQ\ discord (Yang \textit{et al}.~already observed
that the protocol gives an operational interpretation of the CEMI
\cite{YHW08}). The review in this section also serves to prepare for the
result and discussion given in Section~\ref{sec:potential-improve}. Along the
way, we also establish optimality for the total quantum communication rate of
the partial state distribution protocol.

The core protocol underlying partial state distribution is point-to-point
quantum state redistribution (QSR) \cite{DY08,YD09}, so we begin by briefly
reviewing that. Recall that the QSR protocol applies to many copies of a
four-party pure state $\psi_{JKLM}$. That is, suppose that a reference
possesses the $J$ systems, a sender systems $KL$, and a receiver systems $M$.
The goal is to transfer the $K$ systems to the receiver using as few noiseless
qubit channels and as little entanglement as possible while maximizing the
fidelity of the reproduced state. This transfer can happen perfectly in the
asymptotic limit of many copies as long as the rate of quantum communication
is at least $\frac{1}{2}I\left(  K;J|M\right)  _{\psi}$, which is half the
conditional mutual information evaluated with respect to a single copy of
$\psi$. That is, the main result of \cite{DY08,YD09} is that there exists a
sequence of encodings $\mathcal{E}_{K^{n}L^{n}X_{n}\rightarrow L^{n}G_{n}}$
and decodings $\mathcal{D}_{G_{n}Y_{n}M^{n}\rightarrow K^{n}M^{n}}$ such that%
\begin{equation}
\lim_{n\rightarrow\infty}\left\Vert \left(  \mathcal{D}_{G_{n}Y_{n}%
M^{n}\rightarrow K^{n}M^{n}}\circ\mathcal{E}_{K^{n}L^{n}X_{n}\rightarrow
L^{n}G_{n}}\right)  \left(  \psi_{JKLM}^{\otimes n}\otimes\Phi_{X_{n}Y_{n}%
}\right)  -\psi_{JKLM}^{\otimes n}\right\Vert _{1}=0,
\label{eq:state-redist-performance}%
\end{equation}
where $\Phi_{X_{n}Y_{n}}$ is a maximally entangled state and%
\begin{equation}
\lim_{n\rightarrow\infty}\frac{1}{n}\log\dim\left(  G_{n}\right)  =\frac{1}%
{2}I\left(  K;J|M\right)  _{\psi}.
\end{equation}
The QSR protocol can also generate entanglement, but we are not concerned with
this aspect in what follows.

For convenience of presentation, we will review the partial state distribution
protocol for the case of four parties (one central sender and three
receivers), with it being clear how to extend the idea to more parties. Let
the state of interest be $\rho_{A_{1}A_{1}^{\prime}A_{2}A_{2}^{\prime}%
A_{3}A_{3}^{\prime}}$ and let $\phi_{A_{1}A_{1}^{\prime}A_{2}A_{2}^{\prime
}A_{3}A_{3}^{\prime}R}$ be a purification of it. Partial state distribution
begins with the central sender possessing the systems $RA_{1}A_{2}A_{3}$,
Receiver~1 system $A_{1}^{\prime}$, Receiver~2 system $A_{2}^{\prime}$, and
Receiver~3 system $A_{3}^{\prime}$. We assume that the central sender shares
unlimited entanglement with each of the receivers before communication begins.
The partial state distribution protocol gives an operational interpretation of
the information quantity%
\begin{equation}
I\left(  A_{1}A_{1}^{\prime}:A_{2}A_{2}^{\prime}:A_{3}A_{3}^{\prime}\right)
_{\rho}-I\left(  A_{1}^{\prime}:A_{2}^{\prime}:A_{3}^{\prime}\right)  _{\rho
}\label{eq:PSD-info-gap}%
\end{equation}
as twice the total rate of quantum communication needed by the central sender
in order to transfer the system $A_{1}$ to Receiver$~$1, $A_{2}$ to
Receiver~2, and $A_{3}$ to Receiver~3. In order to see this, consider that
Lemma~\ref{lem:chain-for-gap}\ gives the following expansion:%
\begin{multline}
I\left(  A_{1}A_{1}^{\prime}:A_{2}A_{2}^{\prime}:A_{3}A_{3}^{\prime}\right)
_{\rho}-I\left(  A_{1}^{\prime}:A_{2}^{\prime}:A_{3}^{\prime}\right)  _{\rho
}\label{eq:expand-three}\\
=I\left(  A_{3}:A_{1}A_{2}A_{1}^{\prime}A_{2}^{\prime}|A_{3}^{\prime}\right)
_{\rho}+I\left(  A_{2}:A_{1}A_{1}^{\prime}A_{3}^{\prime}|A_{3}^{\prime
}\right)  _{\rho}+I\left(  A_{1}:A_{2}^{\prime}A_{3}^{\prime}|A_{1}^{\prime
}\right)  _{\rho}%
\end{multline}
This suggests that we can perform the QSR protocol three times. Indeed, the
partial state distribution protocol proceeds as follows and as depicted in
Figure~\ref{fig:PSD}:%
\begin{figure}
[ptb]
\begin{center}
\includegraphics[
natheight=5.027200in,
natwidth=12.480100in,
height=2.4535in,
width=6.0502in
]%
{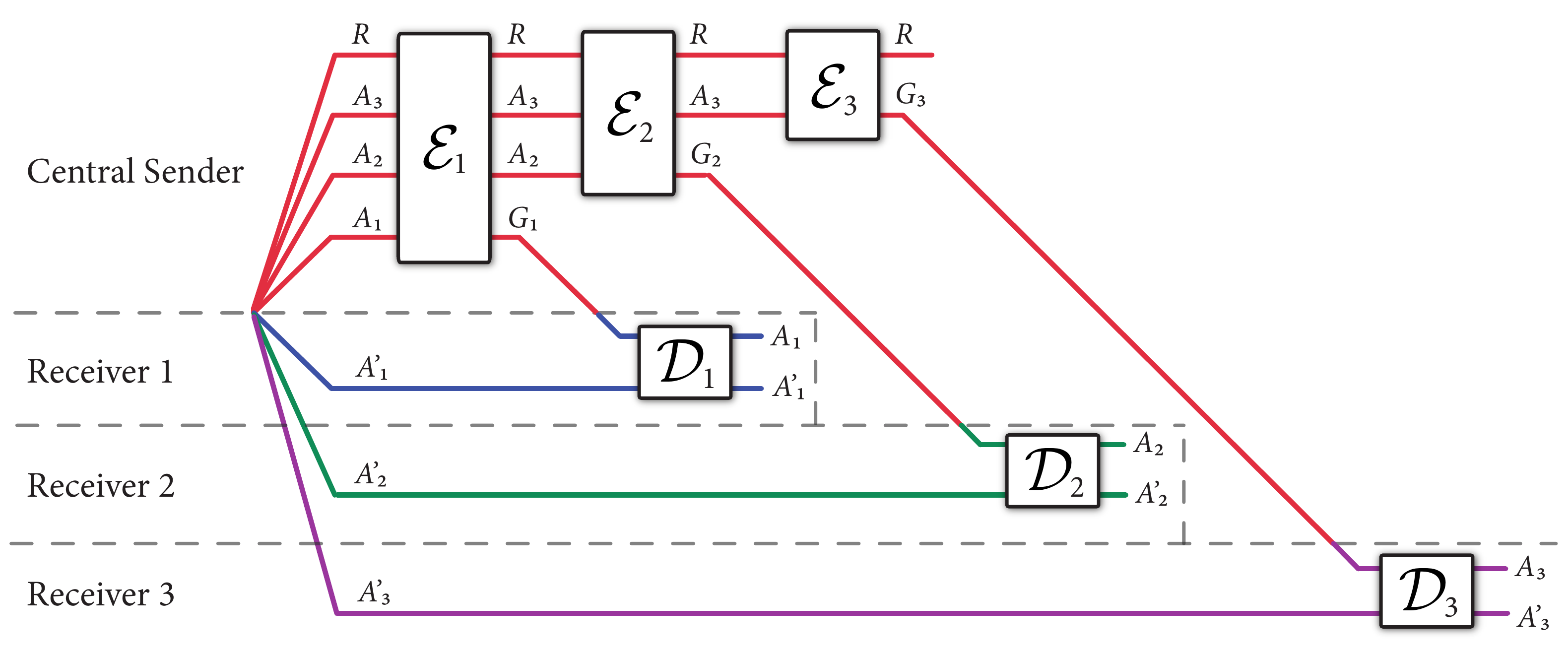}%
\caption{\textbf{Partial state distribution \cite{YHW08}.} The partial state
distribution protocol is a way for a central sender to communicate various
systems to receivers possessing quantum side information. The protocol makes
use of quantum state redistribution for each round (in the figure, for
simplicity, we do not depict the entanglement consumed or generated by the
protocol and it is also implicit that the protocol acts on many copies of the
initial state). The total quantum communication cost is given by $I\left(
A_{1}A_{1}^{\prime}:A_{2}A_{2}^{\prime}:A_{3}A_{3}^{\prime}\right)  -I\left(
A_{1}^{\prime}:A_{2}^{\prime}:A_{3}^{\prime}\right)  $. Observe that the
decodings can proceed in any order. The generalization of this protocol to
more parties is straightforward.}%
\label{fig:PSD}%
\end{center}
\end{figure}

\begin{enumerate}
\item The first round corresponds to the term $I\left(  A_{1}:A_{2}^{\prime
}A_{3}^{\prime}|A_{1}^{\prime}\right)  _{\rho}$. The central sender begins
with systems $RA_{1}A_{2}A_{3}$, Receiver~1 has system $A_{1}^{\prime}$, and
the other receivers have systems $A_{2}^{\prime}A_{3}^{\prime}$ (which play
the role of reference systems in the point-to-point QSR\ protocol). The sender
acts with an encoding $\mathcal{E}_{R^{n}A_{1}^{n}A_{2}^{n}A_{3}^{n}X_{n}%
^{1}\rightarrow R^{n}A_{2}^{n}A_{3}^{n}G_{n}^{1}}^{1}$ and transmits system
$G_{n}^{1}$ to the receiver. Receiver~1 then performs a decoding
$\mathcal{D}_{G_{n}^{1}Y_{n}^{1}A_{1}^{\prime n}\rightarrow A_{1}^{n}%
A_{1}^{\prime n}}^{1}$ to recover the $A_{1}^{\prime}$ systems.

\item The second round corresponds to $I\left(  A_{2}:A_{1}A_{1}^{\prime}%
A_{3}^{\prime}|A_{3}^{\prime}\right)  _{\rho}$. The central sender begins with
systems $RA_{2}A_{3}$, Receiver~2 has system $A_{2}^{\prime}$, and the other
receivers have systems $A_{1}A_{1}^{\prime}A_{3}^{\prime}$ (these systems now
play the role of reference systems in the point-to-point QSR\ protocol). The
central sender acts with an encoding $\mathcal{E}_{R^{n}A_{2}^{n}A_{3}%
^{n}X_{n}^{2}\rightarrow R^{n}A_{3}^{n}G_{n}^{2}}^{2}$ and transmits system
$G_{n}^{2}$ to the receiver. Receiver~2 then performs a decoding
$\mathcal{D}_{G_{n}^{2}Y_{n}^{2}A_{2}^{\prime n}\rightarrow A_{2}^{n}%
A_{2}^{\prime n}}^{2}$ to recover the $A_{2}^{\prime}$ systems.

\item The third round corresponds to $I\left(  A_{3}:A_{1}A_{2}A_{1}^{\prime
}A_{2}^{\prime}|A_{3}^{\prime}\right)  _{\rho}$. The central sender begins
with systems $RA_{3}$, Receiver~3 has system $A_{3}^{\prime}$, and the other
receivers have systems $A_{1}A_{1}^{\prime}A_{2}A_{2}^{\prime}$ (these now
playing the role of reference systems in the point-to-point QSR\ protocol).
The central sender acts with an encoding $\mathcal{E}_{R^{n}A_{3}^{n}X_{n}%
^{3}\rightarrow R^{n}G_{n}^{3}}^{3}$ and transmits system $G_{n}^{3}$ to
Receiver~3. Receiver~3 then performs a decoding $\mathcal{D}_{G_{n}^{3}%
Y_{n}^{3}A_{3}^{\prime n}\rightarrow A_{3}^{n}A_{3}^{\prime n}}$ to recover
the $A_{3}^{\prime}$ systems.
\end{enumerate}

Since all three protocols perform perfectly in the asymptotic limit, by
exploiting the triangle inequality with (\ref{eq:state-redist-performance})
three times, we find that%
\begin{equation}
\lim_{n\rightarrow\infty}\left\Vert \left(  \mathcal{D}_{n}^{3}\circ
\mathcal{E}_{n}^{3}\circ\mathcal{D}_{n}^{2}\circ\mathcal{E}_{n}^{2}%
\circ\mathcal{D}_{n}^{1}\circ\mathcal{E}_{n}^{1}\right)  \left(  \phi
_{A_{1}A_{1}^{\prime}A_{2}A_{2}^{\prime}A_{3}A_{3}^{\prime}R}^{\otimes
n}\otimes\bigotimes\limits_{i=1}^{3}\Phi_{X_{n}^{i}Y_{n}^{i}}\right)
-\phi_{A_{1}A_{1}^{\prime}A_{2}A_{2}^{\prime}A_{3}A_{3}^{\prime}R}^{\otimes
n}\right\Vert _{1}=0,
\end{equation}
with%
\begin{multline}
2\lim_{n\rightarrow\infty}\frac{1}{n}\log\left(  \dim\left(  G_{n}^{1}\right)
\dim\left(  G_{n}^{2}\right)  \dim\left(  G_{n}^{3}\right)  \right)
\label{eq:register-size}\\
=I\left(  A_{3}:A_{1}A_{2}A_{1}^{\prime}A_{2}^{\prime}|A_{3}^{\prime}\right)
_{\rho}+I\left(  A_{2}:A_{1}A_{1}^{\prime}A_{3}^{\prime}|A_{3}^{\prime
}\right)  _{\rho}+I\left(  A_{1}:A_{2}^{\prime}A_{3}^{\prime}|A_{1}^{\prime
}\right)  _{\rho}%
\end{multline}
Due to the nature of this protocol, observe that we can commute all of the
decoding maps to the end and each of these decodings commute with each other
since they act on different spaces. That is, we have that%
\begin{equation}
\lim_{n\rightarrow\infty}\left\Vert \left(  \left(  \mathcal{D}_{n}^{3}%
\otimes\mathcal{D}_{n}^{2}\otimes\mathcal{D}_{n}^{1}\right)  \circ
\mathcal{E}_{n}^{3}\circ\mathcal{E}_{n}^{2}\circ\mathcal{E}_{n}^{1}\right)
\left(  \phi_{A_{1}A_{1}^{\prime}A_{2}A_{2}^{\prime}A_{3}A_{3}^{\prime}%
R}^{\otimes n}\otimes\bigotimes\limits_{i=1}^{3}\Phi_{X_{n}^{i}Y_{n}^{i}%
}\right)  -\phi_{A_{1}A_{1}^{\prime}A_{2}A_{2}^{\prime}A_{3}A_{3}^{\prime}%
R}^{\otimes n}\right\Vert _{1}=0. \label{eq:state-redist-concat}%
\end{equation}
(We cannot however commute the encodings with each other.)

An interesting observation from \cite{YHW08}\ is that the information quantity
in (\ref{eq:expand-three}) is conservative, corresponding to the different
expansions in Lemma~\ref{lem:chain-for-gap}\ and, operationally, to the fact
that we can perform the partial state distribution protocol in any order (we
would however require different encodings and decodings in order to do so).
Also, \cite{YHW08}\ interpreted the CEMI\ in terms of the partial state
distribution protocol as the total rate of quantum communication needed to
transfer the systems $A_{1}$ through $A_{3}$ to independent receivers who
possess the best possible quantum side information in the form of extension
systems $A_{1}^{\prime}$, $A_{2}^{\prime}$, and $A_{3}^{\prime}$, generalizing
the squashed entanglement interpretation from \cite{O08}\ to the multipartite setting.

\subsection{Optimality}

The optimality of the total quantum communication rate in partial state
distribution was not discussed in \cite{YHW08}, but it follows from a simple
argument that exploits the structure of any protocol for partial state
distribution and a few salient properties of the multipartite information. A
proof proceeds similarly to \cite[Theorem~13]{WDHW13}. Indeed, any general
protocol for partial state distribution has the form given in
Figure~\ref{fig:PSD}, with the exception that the encoder can be taken as just
one CPTP\ linear map from the input systems $R^{n}A_{1}^{n}A_{2}^{n}A_{3}%
^{n}X_{1}^{n}X_{2}^{n}X_{3}^{n}$ to the systems $R^{n}G_{1}^{n}G_{2}^{n}%
G_{3}^{n}$. Let $\sigma$ denote the global state after the encoder acts. A
protocol for partial state distribution has a final state $\omega$\ after the
local decodings which is $\varepsilon$-close in trace distance to the ideal
i.i.d.~state $\phi_{A_{1}A_{1}^{\prime}A_{2}A_{2}^{\prime}A_{3}A_{3}^{\prime
}R}^{\otimes n}$. So we proceed with the following chain of inequalities%
\begin{align}
nI\left(  A_{1}A_{1}^{\prime}:A_{2}A_{2}^{\prime}:A_{3}A_{3}^{\prime}\right)
_{\phi}  & =I\left(  A_{1}^{n}A_{1}^{\prime n}:A_{2}^{n}A_{2}^{\prime n}%
:A_{3}^{n}A_{3}^{\prime n}\right)  _{\phi^{\otimes n}}\\
& \leq I\left(  A_{1}^{n}A_{1}^{\prime n}:A_{2}^{n}A_{2}^{\prime n}:A_{3}%
^{n}A_{3}^{\prime n}\right)  _{\omega}+f\left(  \varepsilon\right)  .
\end{align}
The first equality is from the additivity of the multipartite information on
tensor-power states and the inequality follows from the assumption that
$\omega$ is $\varepsilon$-close to the ideal state and by applying the
Fannes-Audenaert inequality \cite{A07}\ with $f\left(  \varepsilon\right)  $ a
function with the property that $\lim_{\varepsilon\rightarrow0}\lim
_{n\rightarrow\infty}\frac{1}{n}f\left(  \varepsilon\right)  =0$. Continuing,
we have that%
\begin{align}
I\left(  A_{1}^{n}A_{1}^{\prime n}:A_{2}^{n}A_{2}^{\prime n}:A_{3}^{n}%
A_{3}^{\prime n}\right)  _{\omega}  & \leq I\left(  A_{1}^{\prime
n}G_{1}^{n}Y_{1}^{n}:A_{2}^{\prime n}G_{2}^{n}Y_{2}^{n}:A_{3}^{\prime n}G_{3}^{n}
Y_{3}^{n}\right)  _{\sigma}\\
& \leq I\left(  A_{1}^{\prime n}Y_{1}^{n}:A_{2}^{\prime n}Y_{2}^{n}%
:A_{3}^{\prime n}Y_{3}^{n}\right)  _{\sigma}+2\log\left(  \left\vert G_{1}%
^{n}\right\vert \left\vert G_{2}^{n}\right\vert \left\vert G_{3}%
^{n}\right\vert \right)  \\
& =I\left(  A_{1}^{\prime n}:A_{2}^{\prime n}:A_{3}^{\prime n}\right)
_{\sigma}+2\log\left(  \left\vert G_{1}^{n}\right\vert \left\vert G_{2}%
^{n}\right\vert \left\vert G_{3}^{n}\right\vert \right)  \label{eq:third-eq}%
\\
& =nI\left(  A_{1}^{\prime}:A_{2}^{\prime}:A_{3}^{\prime}\right)  _{\phi
}+2\log\left(  \left\vert G_{1}^{n}\right\vert \left\vert G_{2}^{n}\right\vert
\left\vert G_{3}^{n}\right\vert \right)  ,
\end{align}
where the first inequality follows from quantum data processing (the local
decoders can only decrease the multipartite information). The second
inequality follows from%
\begin{align}
I\left(  C_{1}D_{1}:\cdots:C_{l}D_{l}\right)    & =\sum_{i=1}^{l}H\left(
C_{i}D_{i}\right)  -H\left(  C_{1}D_{1}\cdots C_{l}D_{l}\right)  \\
& =\sum_{i=1}^{l}H\left(  C_{i}\right)  +H\left(  D_{i}|C_{i}\right)
-H\left(  C_{1}\cdots C_{l}\right)  -H\left(  D_{1}\cdots D_{l}|C_{1}\cdots
C_{l}\right)  \\
& \leq\sum_{i=1}^{l}H\left(  C_{i}\right)  -H\left(  C_{1}\cdots C_{l}\right)
+2\log\left(  \left\vert D_{1}\right\vert \times\cdots\times\left\vert
D_{l}\right\vert \right)  \\
& =I\left(  C_{1}:\cdots:C_{l}\right)  +2\log\left(  \left\vert D_{1}%
\right\vert \times\cdots\times\left\vert D_{l}\right\vert \right)
\end{align}
The equality in (\ref{eq:third-eq}) follows because the systems are product
with respect to the cut $A_{1}^{\prime n}A_{2}^{\prime n}A_{3}^{\prime
n}|Y_{1}^{n}|Y_{2}^{n}|Y_{3}^{n}$. The final equality is again additivity.
Putting everything together we find that%
\begin{equation}
\frac{1}{2}\left[  I\left(  A_{1}A_{1}^{\prime}:A_{2}A_{2}^{\prime}:A_{3}%
A_{3}^{\prime}\right)  _{\phi}-I\left(  A_{1}^{\prime}:A_{2}^{\prime}%
:A_{3}^{\prime}\right)  _{\phi}\right]  \leq\frac{1}{n}\log\left(  \left\vert
G_{1}^{n}\right\vert \left\vert G_{2}^{n}\right\vert \left\vert G_{3}%
^{n}\right\vert \right)  +\frac{1}{n}f\left(  \varepsilon\right)  .
\end{equation}
Taking the limit as $n\rightarrow\infty$ and $\varepsilon\rightarrow0$ then
establishes the information gap in (\ref{eq:PSD-info-gap}) as twice the
minimum total rate of quantum communication needed in any partial state
distribution protocol. This analysis clearly extends to any finite number of parties.

\subsection{Operational interpretation of the MSQ\ discord}

\label{sec:op-int}The partial state distribution protocol gives a compelling
operational interpretation of the MSQ\ discord, different from and arguably
simpler than those considered in previous contexts \cite{CABMPW11,MD11}.
Suppose that we have a multipartite state $\rho_{A_{1}\cdots A_{l}}$ shared by
$l$ local parties, each of whom possesses system $A_{i}$ where $i\in\left\{
1,\ldots,l\right\}  $. Let $\phi_{RA_{1}\cdots A_{l}}$ be a state which
purifies $\rho_{A_{1}\cdots A_{l}}$, where $R$ is an environment system
inaccessible to the local parties. Suppose now that a measurement occurs on
each of the systems, according to the measurement maps $\mathcal{M}%
_{A_{1}\rightarrow X_{1}}^{1},\ldots,\mathcal{M}_{A_{l}\rightarrow X_{l}}^{l}%
$, producing the state $\omega_{RX_{1}\cdots X_{l}}$:%
\begin{equation}
\omega_{RX_{1}\cdots X_{l}}\equiv\left(  \mathcal{M}_{A_{1}\rightarrow X_{1}%
}^{1}\otimes\cdots\otimes\mathcal{M}_{A_{l}\rightarrow X_{l}}^{l}\right)
\left(  \phi_{RA_{1}\cdots A_{l}}\right)  .
\end{equation}
A measurement corresponds to a loss of information, and one way to represent
this is with isometric extensions of the measurement process, so that the full
state is%
\begin{equation}
\omega_{RX_{1}E_{1}\cdots X_{l}E_{l}}\equiv\left(  \mathcal{U}_{A_{1}%
\rightarrow X_{1}E_{1}}^{\mathcal{M}^{1}}\otimes\cdots\otimes\mathcal{U}%
_{A_{l}\rightarrow X_{l}E_{l}}^{\mathcal{M}^{l}}\right)  \left(  \phi
_{RA_{1}\cdots A_{l}}\right)  .
\end{equation}
and $\mathcal{U}_{A_{i}\rightarrow X_{i}E_{i}}^{\mathcal{M}^{i}}$ is an
isometric extension of the measurement map $\mathcal{M}_{A_{i}\rightarrow
X_{i}}^{i}$. Since the systems $E_{1}$, \ldots, $E_{l}$ are lost to the
environment after the measurement process, it becomes the case that the
environment possesses the systems $R$, $E_{1}$, \ldots, $E_{l}$, and each of
the local parties possesses one of the measurement outcomes.

With this setup, we can now see that the (unoptimized)\ MSQ\ discord%
\begin{equation}
I\left(  A_{1}:\cdots:A_{l}\right)  _{\rho}-I\left(  X_{1}:\cdots
:X_{l}\right)  _{\omega}=I\left(  X_{1}E_{1}:\cdots:X_{l}E_{l}\right)  _{\rho
}-I\left(  X_{1}:\cdots:X_{l}\right)  _{\omega}%
\end{equation}
is equal to the twice the total rate of quantum communication needed for the
environment to send the systems $E_{1}$, \ldots, $E_{l}$ back to each of the
local parties in order to restore the coherence lost in the measurement
processes. Due to the fact that the QSR\ protocol is dual under time reversal
\cite{DY08,YD09}, the unoptimized MSQ\ discord is also equal to the twice the
total rate of quantum communication needed by the local parties to transmit
the systems $E_{1}$, \ldots, $E_{l}$ back to the environment, thus
additionally characterizing the rate at which coherence is lost in the
measurement process. The (optimized)\ MSQ\ discord simply includes a further
optimization over the measurements themselves in order to minimize the total
quantum communication cost.\footnote{A subtle point here is that one could
more generally include an optimization over collective quantum measurements
acting on many copies of the state, which would result in a
\textit{regularized} MSQ\ discord being equal to the total quantum
communication cost.}

\subsection{Operational interpretation of the quantum discord}

We remark that this approach in terms of partial state distribution gives as a
special case a compelling operational interpretation of the original quantum
discord, again different from and arguably simpler than those considered
previously \cite{CABMPW11,MD11}. Indeed, consider a bipartite state $\rho
_{AB}$ and a measurement map $\mathcal{M}_{A\rightarrow X}$. The unoptimized
quantum discord is defined as%
\begin{align}
I\left(  A;B\right)  _{\rho}-I\left(  X;B\right)  _{\mathcal{M}\left(
\rho\right)  }  &  =I\left(  XE;B\right)  _{\mathcal{U}\left(  \rho\right)
}-I\left(  X;B\right)  _{\mathcal{U}\left(  \rho\right)  }\\
&  =I\left(  E;B|X\right)  _{\mathcal{U}\left(  \rho\right)  },
\end{align}
where the first equality follows because every measurement map has an
isometric extension $\mathcal{U}_{A\rightarrow XE}^{\mathcal{M}}$ and the
mutual information is invariant under local isometries. The second equality is
a consequence of the chain rule (this rewriting of discord in terms of
conditional mutual information was first explicitly given in \cite{P12}).
Purifying the original state with a reference system $R$, we have a pure state
on systems $REXB$. After the measurement occurs, it is natural to associate
the system $E$ as being \textquotedblleft lost\textquotedblright\ and thus
given to the other environment system $R$. That is, after the measurement
occurs, the systems $R$ and $E$ are with the environment, the system $X$ is
with a party who has the measurement outcome, and the system $B$ is with
another party who plays no role in the protocol. We can then readily see from
the QSR\ protocol that $I\left(  E;B|X\right)  _{\mathcal{U}\left(
\rho\right)  }$ is twice the rate of quantum communication needed in order to
transmit the system $E$ to the party possessing $X$ (assuming that the $B$
system is with a different party who does not play a role in this transfer).
We can thus interpret $I\left(  E;B|X\right)  _{\mathcal{U}\left(
\rho\right)  }$ as twice the quantum communication cost needed to restore the
coherence that was lost in the measurement process. After this transfer
occurs, $I\left(  E;B|X\right)  _{\mathcal{U}\left(  \rho\right)  }$ is also
equal to twice the quantum communication rate needed to send the system $E$
back to the environment (this is because state redistribution is dual under
time reversal \cite{DY08,YD09}). The quantity thus also characterizes the
amount of quantum information lost in the measurement process. Optimizing over
all measurements gives the optimized discord (keeping in mind that one could
potentially optimize over collective measurements and get a
\textit{regularized} discord).

\section{Potential improvement of the local recoverability inequality}

\label{sec:potential-improve}It might be possible to improve upon the local
recoverability inequality given in (\ref{eq:local-recover}). Here we provide
what might be a first step, which follows an approach recently given in
\cite{BHOS14}.

\begin{proposition}
\label{prop:new-lower}Let $\rho_{A_{1}A_{1}^{\prime}\cdots A_{l}A_{l}^{\prime
}}$ be a multipartite quantum state. Then the following inequality holds%
\begin{multline}
I\left(  A_{1}A_{1}^{\prime}:\cdots:A_{l}A_{l}^{\prime}\right)  _{\rho
}-I\left(  A_{1}^{\prime}:\cdots:A_{l}^{\prime}\right)  _{\rho}%
\label{eq:multipartite-recovery-regularized}\\
\geq\lim_{n\rightarrow\infty}\min_{\mathcal{R}^{1},\cdots\mathcal{R}^{l}}%
\frac{1}{n}D\left(  \rho_{A_{1}A_{1}^{\prime}\cdots A_{l}A_{l}^{\prime}%
}^{\otimes n}\Vert\mathcal{R}_{A_{1}^{\prime n}\rightarrow A_{1}^{n}%
A_{1}^{\prime n}}^{1}\otimes\cdots\otimes\mathcal{R}_{A_{l}^{\prime
n}\rightarrow A_{l}^{n}A_{l}^{\prime n}}^{l}\left(  \rho_{A_{1}^{\prime}\cdots
A_{l}^{\prime}}^{\otimes n}\right)  \right)  ,
\end{multline}
where $\mathcal{R}_{A_{1}^{\prime n}\rightarrow A_{1}^{n}A_{1}^{\prime n}}%
^{1}$, \ldots, $\mathcal{R}_{A_{l}^{\prime n}\rightarrow A_{l}^{n}%
A_{l}^{\prime n}}^{l}$ are a sequence of local recovery maps.
\end{proposition}

\begin{proof}
The proof of this lemma is very similar to the proof of Proposition~3 of
\cite{BHOS14}, except that we invoke the partial state redistribution protocol
reviewed in Section~\ref{sec:PSD}. Picking up from the notation there, and
specializing to a state on systems $A_{1}A_{1}^{\prime}A_{2}A_{2}^{\prime
}A_{3}A_{3}^{\prime}$, let%
\begin{multline}
\varphi_{R^{n}G_{n}^{1}G_{n}^{2}G_{n}^{3}Y_{n}^{1}Y_{n}^{2}Y_{n}^{3}%
A_{1}^{\prime n}A_{2}^{\prime n}A_{3}^{\prime n}}\\
\equiv\mathcal{E}_{R^{n}A_{3}^{n}X_{n}^{3}\rightarrow R^{n}G_{n}^{3}}%
^{3}\left(  \mathcal{E}_{R^{n}A_{2}^{n}A_{3}^{n}X_{n}^{2}\rightarrow
R^{n}A_{3}^{n}G_{n}^{2}}^{2}\left(  \mathcal{E}_{R^{n}A_{1}^{n}A_{2}^{n}%
A_{3}^{n}X_{n}^{1}\rightarrow R^{n}A_{2}^{n}A_{3}^{n}G_{n}^{1}}^{1}\left(
\phi_{A_{1}A_{1}^{\prime}A_{2}A_{2}^{\prime}A_{3}A_{3}^{\prime}R}^{\otimes
n}\otimes\bigotimes\limits_{i=1}^{3}\Phi_{X_{n}^{i}Y_{n}^{i}}\right)  \right)
\right)
\end{multline}
denote the state after the encodings. Tracing over $R^{n}$ and applying the
operator inequality $\sigma_{CD}\leq\left[  \dim\left(  D\right)  \right]
^{2}\tau_{C}\otimes\rho_{D}$ three times, we find that%
\begin{multline}
\varphi_{G_{n}^{1}G_{n}^{2}G_{n}^{3}Y_{n}^{1}Y_{n}^{2}Y_{n}^{3}A_{1}^{\prime
n}A_{2}^{\prime n}A_{3}^{\prime n}}\\
\leq\left[  \dim\left(  G_{n}^{1}\right)  \right]  ^{2}\left[  \dim\left(
G_{n}^{2}\right)  \right]  ^{2}\left[  \dim\left(  G_{n}^{3}\right)  \right]
^{2}\tau_{G_{n}^{1}}\otimes\tau_{G_{n}^{2}}\otimes\tau_{G_{n}^{3}}\otimes
\tau_{Y_{n}^{1}}\otimes\tau_{Y_{n}^{2}}\otimes\tau_{Y_{n}^{3}}\otimes
\rho_{A_{1}^{\prime n}A_{2}^{\prime n}A_{3}^{\prime n}}.
\end{multline}
Now for $i\in\left\{  1,2,3\right\}  $, define the perturbed decoding
operations%
\begin{equation}
\mathcal{\tilde{D}}_{n}^{i}\equiv\left(  1-2^{-n}\right)  \mathcal{D}_{n}%
^{i}+2^{-n}\Lambda_{\text{dep}},
\end{equation}
where $\Lambda_{\text{dep}}$ is the completely depolarizing channel. Since
these are completely positive and acting on different spaces, we find that%
\begin{multline}
\left(  \mathcal{\tilde{D}}_{n}^{3}\otimes\mathcal{\tilde{D}}_{n}^{2}%
\otimes\mathcal{\tilde{D}}_{n}^{1}\right)  \left(  \varphi_{G_{n}^{1}G_{n}%
^{2}G_{n}^{3}Y_{n}^{1}Y_{n}^{2}Y_{n}^{3}A_{1}^{\prime n}A_{2}^{\prime n}%
A_{3}^{\prime n}}\right) \\
\leq\left[  \dim\left(  G_{n}^{1}\right)  \right]  ^{2}\left[  \dim\left(
G_{n}^{2}\right)  \right]  ^{2}\left[  \dim\left(  G_{n}^{3}\right)  \right]
^{2}\left(  \mathcal{R}_{A_{1}^{\prime n}\rightarrow A_{1}^{n}A_{1}^{\prime
n}}^{1}\otimes\mathcal{R}_{A_{2}^{\prime n}\rightarrow A_{2}^{n}A_{2}^{\prime
n}}^{2}\otimes\mathcal{R}_{A_{3}^{\prime n}\rightarrow A_{3}^{n}A_{3}^{\prime
n}}^{3}\right)  \left(  \rho_{A_{1}^{\prime n}A_{2}^{\prime n}A_{3}^{\prime
n}}\right)  ,
\end{multline}
where the recovery map $\mathcal{R}_{A_{i}^{\prime n}\rightarrow A_{i}%
^{n}A_{i}^{\prime n}}^{i}$ for $i\in\left\{  1,2,3\right\}  $ is defined to be
the map that first tensors in maximally mixed states on systems $G_{n}^{i}$
and $Y_{n}^{i}$ and then performs $\mathcal{\tilde{D}}_{n}^{i}$. Using
operator monotonicity of the logarithm, we find that%
\begin{multline}
D\left(  \rho_{A_{1}A_{1}^{\prime}A_{2}A_{2}^{\prime}A_{3}A_{3}^{\prime}%
}^{\otimes n}\Vert\left(  \mathcal{R}_{A_{1}^{\prime n}\rightarrow A_{1}%
^{n}A_{1}^{\prime n}}^{1}\otimes\mathcal{R}_{A_{2}^{\prime n}\rightarrow
A_{2}^{n}A_{2}^{\prime n}}^{2}\otimes\mathcal{R}_{A_{3}^{\prime n}\rightarrow
A_{3}^{n}A_{3}^{\prime n}}^{3}\right)  \left(  \rho_{A_{1}^{\prime n}%
A_{2}^{\prime n}A_{3}^{\prime n}}\right)  \right) \\
\leq D\left(  \rho_{A_{1}A_{1}^{\prime}A_{2}A_{2}^{\prime}A_{3}A_{3}^{\prime}%
}^{\otimes n}\Vert\left(  \mathcal{\tilde{D}}_{n}^{3}\otimes\mathcal{\tilde
{D}}_{n}^{2}\otimes\mathcal{\tilde{D}}_{n}^{1}\right)  \left(  \varphi
_{G_{n}^{1}G_{n}^{2}G_{n}^{3}Y_{n}^{1}Y_{n}^{2}Y_{n}^{3}A_{1}^{\prime n}%
A_{2}^{\prime n}A_{3}^{\prime n}}\right)  \right) \\
+2\log\left(  \dim\left(  G_{n}^{1}\right)  \dim\left(  G_{n}^{2}\right)
\dim\left(  G_{n}^{3}\right)  \right)
\end{multline}
Theorem~3 of \cite{AE05} gives that%
\begin{equation}
\lim_{n\rightarrow\infty}\frac{1}{n}D\left(  \rho_{A_{1}A_{1}^{\prime}%
A_{2}A_{2}^{\prime}A_{3}A_{3}^{\prime}}^{\otimes n}\Vert\left(
\mathcal{\tilde{D}}_{n}^{3}\otimes\mathcal{\tilde{D}}_{n}^{2}\otimes
\mathcal{\tilde{D}}_{n}^{1}\right)  \left(  \varphi_{G_{n}^{1}G_{n}^{2}%
G_{n}^{3}Y_{n}^{1}Y_{n}^{2}Y_{n}^{3}A_{1}^{\prime n}A_{2}^{\prime n}%
A_{3}^{\prime n}}\right)  \right)  =0
\end{equation}
as a consequence of (\ref{eq:state-redist-concat}). With this, we can conclude
the statement in (\ref{eq:multipartite-recovery-regularized}) by combining the
above with (\ref{eq:register-size})\ and (\ref{eq:expand-three}).

It should be clear from here how the general multiparty case proceeds. Letting
the number of parties be some positive integer $l$, we first apply
Lemma~\ref{lem:chain-for-gap}. Next we perform the partial state distribution
protocol in the same fashion as above. Importantly, all of the encodings take
place in a particular order, but the decodings all act on different spaces and
thus commute. Finally, we apply the same reasoning at the end to conclude the
general statement of the lemma.
\end{proof}

\bigskip

\textit{We leave as an open question whether the following inequality holds:}%
\begin{multline}
I\left(  A_{1}A_{1}^{\prime}:\cdots:A_{l}A_{l}^{\prime}\right)  _{\rho
}-I\left(  A_{1}^{\prime}:\cdots:A_{l}^{\prime}\right)  _{\rho}%
\label{eq:extend-FR}\\
\geq-\log F\left(  \rho_{A_{1}A_{1}^{\prime}\cdots A_{l}A_{l}^{\prime}%
},\left(  \mathcal{R}_{A_{1}^{\prime}\rightarrow A_{1}A_{1}^{\prime}}%
^{1}\otimes\cdots\otimes\mathcal{R}_{A_{l}^{\prime}\rightarrow A_{l}%
A_{l}^{\prime}}^{l}\right)  \left(  \rho_{A_{1}^{\prime}\cdots A_{l}^{\prime}%
}\right)  \right)  ,
\end{multline}
\textit{where }$\rho_{A_{1}A_{1}^{\prime}\cdots A_{l}A_{l}^{\prime}}$\textit{
is a multipartite quantum state and }$\mathcal{R}_{A_{1}^{\prime}\rightarrow
A_{1}A_{1}^{\prime}}^{1}$\textit{, \ldots, }$\mathcal{R}_{A_{l}^{\prime
}\rightarrow A_{l}A_{l}^{\prime}}^{l}$\textit{ are some local recovery maps.
}At the very least, the inequality holds for classical systems as a
consequence of Theorem~5 of \cite{LW14}. By extending the methods of
\cite{FR14,BHOS14}, it might be possible to establish the above inequality.

\section{Discussion}

\label{sec:conclusion}We have demonstrated how the inequality in (\ref{eq:FR})
implies a relation between the multipartite information gap $I\left(
A_{1}A_{1}^{\prime}:\cdots:A_{l}A_{l}^{\prime}\right)  -I\left(  A_{1}%
^{\prime}:\cdots:A_{l}^{\prime}\right)  $ and local recoverability. Namely, a
multipartite state has a multipartite information gap nearly equal to zero if
and only if the systems $A_{1}$, \ldots, $A_{l}$ are locally recoverable from
the respective systems $A_{1}^{\prime}$, \ldots, $A_{l}^{\prime}$. This result
in turn implies that 1)\ the multipartite symmetric quantum discord of a state
$\rho_{A_{1}\cdots A_{l}}$ is nearly equal to zero if and only if the state is
locally recoverable after measurements\ occur on each of the systems and
2)\ the conditional entanglement of multipartite information is faithful. We
have also given a compelling operational interpretation of the multipartite
symmetric quantum discord as the twice the total quantum communication cost
needed to restore the coherence lost from a sequence of local measurements. A
similar operational interpretation applies to the original quantum discord
quantity as well. Finally, Proposition~\ref{prop:new-lower} gives another
lower bound on the multipartite information gap by generalizing an approach
recently outlined in \cite{BHOS14}.

There are several open questions to consider going forward from here. First,
it would be interesting if the inequality in (\ref{eq:extend-FR}) were true.
It is true for classical systems, and to show it for quantum systems, one
could consider extending the methods given in \cite[Proposition~4]{BHOS14} to
this multipartite setting. Next, in light of the recent developments in
\cite{BSW14,SBW14,SW14}, one could define a geometric CEMI as follows:%
\begin{multline}
E_{I}^{F}\left(  A_{1}:\cdots:A_{l}\right)  _{\rho}\equiv\\
-\frac{1}{2}\log\sup_{\substack{\rho_{A_{1}A_{1}^{\prime}\cdots A_{l}%
A_{l}^{\prime}},\\\mathcal{R}^{1},\cdots,\mathcal{R}^{l}}}F\left(  \rho
_{A_{1}A_{1}^{\prime}\cdots A_{l}A_{l}^{\prime}},\left(  \mathcal{R}%
_{A_{1}^{\prime}\rightarrow A_{1}A_{1}^{\prime}}^{1}\otimes\cdots
\otimes\mathcal{R}_{A_{l}^{\prime}\rightarrow A_{l}A_{l}^{\prime}}^{l}\right)
\left(  \rho_{A_{1}^{\prime}\cdots A_{l}^{\prime}}\right)  \right)
\end{multline}
where the optimization is over all extensions of $\rho_{A_{1}\cdots A_{l}}$
and all recovery maps $\mathcal{R}^{1},\cdots,\mathcal{R}^{l}$. One could also
define a multipartite surprisal of measurement recoverability as%
\begin{equation}
D^{F}\left(  \overline{A_{1}}:\cdots:\overline{A_{l}}\right)  _{\rho}%
\equiv-\log\sup_{\mathcal{E}_{A_{1}}^{1},\cdots,\mathcal{E}_{A_{l}}^{l}%
}F\left(  \rho_{A_{1}\cdots A_{l}},\left(  \mathcal{E}_{A_{1}}^{1}%
\otimes\cdots\otimes\mathcal{E}_{A_{l}}^{l}\right)  \left(  \rho_{A_{1}\cdots
A_{l}}\right)  \right)  ,
\end{equation}
where the optimization is over all local entanglement breaking channels. One
could even consider other discord-like quantities of the above form, but
involving alternate (pseudo-)distance measures such as the trace distance and
relative entropy. We can already conclude that the geometric CEMI\ is faithful
by the results given in this paper, and one could pursue further properties of
these quantities in future work.

\bigskip

\textbf{Acknowledgements.} The author is thankful to Mario Berta, Marco Piani,
Kaushik Seshadreesan, and Andreas Winter for discussions related to the topic
of this paper. The author acknowledges support from startup funds from the
Department of Physics and Astronomy at LSU, the NSF\ under Award
No.~CCF-1350397, and the DARPA Quiness Program through US Army Research Office
award W31P4Q-12-1-0019.

\appendix

\section{Multipartite de Finetti theorem}

We begin by recalling Theorem~II.7' of \cite{CKMR07}:

\begin{theorem}
\label{thm:de-finetti}Let $\zeta_{EF}$ be a $k$-extendible state, in the sense
that there is a state $\theta_{EF_{1}\cdots F_{k}}$ that is invariant with
respect to permutations of the $F$ systems and such that
$\operatorname{Tr}_{F_{2}\cdots F_{k}}\left\{  \theta_{EF_{1}\cdots F_{k}%
}\right\}  =\zeta_{EF}$. Then there exists a measure $d\mu\left(  \sigma
_{F}\right)  $ on states $\sigma_{F}$ on the $F$ system and a family of states
$\left\{  \xi_{E}^{\sigma}\right\}  $ parametrized by $\sigma_{F}$ such that%
\begin{equation}
\left\Vert \zeta_{EF}-\int d\mu\left(  \sigma_{F}\right)  \xi_{E}^{\sigma
}\otimes\sigma_{F}\right\Vert _{1}\leq\frac{2\left\vert F\right\vert^2 }{k}.
\end{equation}

\end{theorem}

The following proposition follows directly from prior results in the literature,
but we state it here and give a brief
proof for readers' convenience:

\begin{proposition}
\label{prop:multi-de-finetti}Let $\rho_{A_{1}A_{2}\cdots A_{l}}$ be a
multipartite $k$-extendible state, i.e., there exists a state%
\begin{equation}
\omega_{A_{1}A_{2,1}\cdots A_{2,k}\cdots A_{l,1}\cdots A_{l,k}}%
\end{equation}
that is permutation invariant with respect to the systems $A_{j,1}\cdots
A_{j,k}$, for each $j\in\left\{  2,\ldots,l\right\}  $, and such that
$\rho_{A_{1}A_{2}\cdots A_{l}}=\operatorname{Tr}_{A_{2,2}\cdots A_{2,k}\cdots
A_{l,2}\cdots A_{l,k}}\{\omega\}$. Then%
\begin{equation}
\left\Vert \rho_{A_{1}A_{2}\cdots A_{l}}-\operatorname{SEP}\left(  A_{1}%
:A_{2}:\cdots:A_{l}\right)  \right\Vert _{1}\leq\frac{2}{k}\left(  \left\vert
A_{2}\right\vert ^{2}+\cdots+\left\vert A_{l}\right\vert ^{2}\right)
.\label{eq:multipart-de-finetti}%
\end{equation}

\end{proposition}

\begin{proof}
The idea is to proceed similar to the proof of \cite[Theorem~1]{DPS05}, but
here invoking Theorem~\ref{thm:de-finetti} several times. We consider a
particular example with only three parties for simplicity, and it will then be
clear how the approach extends to states with more parties. So we begin with a
multipartite $k$-extendible state $\rho_{ABC}$ and its multipartite
$k$-extension $\omega_{AB_{1}\cdots B_{k}C_{1}\cdots C_{k}}$. We first apply
Theorem~\ref{thm:de-finetti} to $\omega_{AB_{1}\cdots B_{k}C}$ (where
$C=C_{1}$), setting $E=AB_{1}\cdots B_{k}$ and $F=C$. We can conclude that
there exists a measure $d\mu\left(  \sigma_{C}\right)  $ and a family of
states $\{\xi_{AB_{1}\cdots B_{k}}^{\sigma}\}$ such that%
\begin{equation}
\left\Vert \omega_{AB_{1}\cdots B_{k}C}-\int d\mu\left(  \sigma_{C}\right)
\xi_{AB_{1}\cdots B_{k}}^{\sigma}\otimes\sigma_{C}\right\Vert _{1}\leq
\frac{2\left\vert C\right\vert ^{2}}{k}.
\end{equation}
Due to the invariance of the state $\omega_{AB_{1}\cdots B_{k}C}$ under
permutations of the $B$ systems and monotonicity of the trace norm under
quantum operations, we can conclude the following inequality%
\begin{equation}
\left\Vert \omega_{AB_{1}\cdots B_{k}C}-\int d\mu\left(  \sigma_{C}\right)
\overline{\xi}_{AB_{1}\cdots B_{k}}^{\sigma}\otimes\sigma_{C}\right\Vert
_{1}\leq\frac{2\left\vert C\right\vert ^{2}}{k},\label{eq:de-finetti-1-iter}%
\end{equation}
where $\overline{\xi}_{AB_{1}\cdots B_{k}}^{\sigma}\equiv\overline{\Pi}%
_{B^{k}}\left(  \xi_{AB_{1}\cdots B_{k}}\right)  $, with $\overline{\Pi
}_{B^{k}}$ a channel that randomly permutes the $B$ systems (defined in
(\ref{eq:randomizing-channel})). Given that each state $\overline{\xi}%
_{AB_{1}\cdots B_{k}}^{\sigma}$ is permutation symmetric with respect to the
$B$ systems, we can again invoke Theorem~\ref{thm:de-finetti} to conclude that
there exists a measure $d\mu\left(  \tau\left(  \sigma\right)  \right)  $ on
states $\tau\left(  \sigma\right)  _{B}$ and a family of states $\{\chi
_{A}^{\tau\left(  \sigma\right)  }\}$ such that%
\begin{equation}
\left\Vert \overline{\xi}_{AB}^{\sigma}-\int d\mu\left(  \tau\left(
\sigma\right)  \right)  \chi_{A}^{\tau\left(  \sigma\right)  }\otimes
\tau\left(  \sigma\right)  _{B}\right\Vert _{1}\leq\frac{2\left\vert
B\right\vert ^{2}}{k}.
\end{equation}
This implies that%
\begin{equation}
\left\Vert \overline{\xi}_{AB}^{\sigma}\otimes\sigma_{C}-\int d\mu\left(
\tau\left(  \sigma\right)  \right)  \chi_{A}^{\tau\left(  \sigma\right)
}\otimes\tau\left(  \sigma\right)  _{B}\otimes\sigma_{C}\right\Vert _{1}%
\leq\frac{2\left\vert B\right\vert ^{2}}{k},
\end{equation}
and applying convexity of the trace norm gives%
\begin{equation}
\left\Vert \int d\mu\left(  \sigma_{C}\right)  \overline{\xi}_{AB}^{\sigma
}\otimes\sigma_{C}-\int\int d\mu\left(  \sigma_{C}\right)  d\mu\left(
\tau\left(  \sigma\right)  \right)  \chi_{A}^{\tau\left(  \sigma\right)
}\otimes\tau\left(  \sigma\right)  _{B}\otimes\sigma_{C}\right\Vert _{1}%
\leq\frac{2\left\vert B\right\vert ^{2}}{k}.\label{eq:triangle-1}%
\end{equation}
Applying monotonicity of the trace norm under partial trace to
(\ref{eq:de-finetti-1-iter}) gives%
\begin{equation}
\left\Vert \rho_{ABC}-\int d\mu\left(  \sigma_{C}\right)  \overline{\xi}%
_{AB}^{\sigma}\otimes\sigma_{C}\right\Vert _{1}\leq\frac{2\left\vert
C\right\vert ^{2}}{k}.\label{eq:triangle-2}%
\end{equation}
We finally combine (\ref{eq:triangle-1}) and (\ref{eq:triangle-2}) with the
triangle inequality to get%
\begin{equation}
\left\Vert \rho_{ABC}-\int\int d\mu\left(  \sigma_{C}\right)  d\mu\left(
\tau\left(  \sigma\right)  \right)  \chi_{A}^{\tau\left(  \sigma\right)
}\otimes\tau\left(  \sigma\right)  _{B}\otimes\sigma_{C}\right\Vert _{1}%
\leq\frac{2}{k}\left(  \left\vert B\right\vert ^{2}+\left\vert C\right\vert
^{2}\right)  .
\end{equation}
Since the state on the right is a convex combination of product states, it is
fully separable, so that we can conclude%
\begin{equation}
\left\Vert \rho_{ABC}-\operatorname{SEP}\left(  A:B:C\right)  \right\Vert
_{1}\leq\frac{2}{k}\left(  \left\vert B\right\vert ^{2}+\left\vert
C\right\vert ^{2}\right)  .
\end{equation}
Extending this proof to more parties is done in the obvious way, so that we can
conclude (\ref{eq:multipart-de-finetti}).
\end{proof}

\bibliographystyle{alpha}
\bibliography{Ref}

\end{document}